\documentclass[11pt,draftcls,peerreviewca,onecolumn,a4paper,dvips]{IEEEtran}
\usepackage{tipa}
\usepackage{amsfonts}
\usepackage{amsfonts}
\usepackage{amsfonts}
\usepackage{amsfonts}
\usepackage{amssymb}
\usepackage{stfloats}
\usepackage{cite}
\usepackage{graphicx}
\usepackage{psfrag}
\usepackage{subfigure}
\usepackage{amsmath}
\usepackage{array}
\usepackage[usenames]{color}

\newcommand{\black}{\color{black}}

\begin{document}

\title{Decentralized Delay Optimal Control for Interference Networks with Limited Renewable Energy Storage}

\newtheorem{Thm}{Theorem}
\newtheorem{Lem}{Lemma}
\newtheorem{Cor}{Corollary}
\newtheorem{Def}{Definition}
\newtheorem{Exam}{Example}
\newtheorem{Alg}{Algorithm}
\newtheorem{Prob}{Problem}
\newtheorem{Rem}{Remark}
\newtheorem{Proof}{Proof}
\newtheorem{Subproblem}{Subproblem}
\newtheorem{assumption}{Assumption}

\author{\authorblockN{Huang Huang {\em Member, IEEE}, Vincent K. N. Lau, {\em Fellow, IEEE}}
\thanks{The authors are with the Department of Electronic and Computer Engineering
(ECE), Hong Kong University of Science and Technology (HKUST), Hong
Kong.}
}

\maketitle

\begin{abstract}
In this paper, we consider delay minimization for interference networks with renewable energy source, where the transmission power of a node comes from both the conventional utility power (AC power) and the renewable energy source. We assume the transmission power of each node is a function of the {\em local channel state, local data queue state and local energy queue state} only. In turn, we consider two delay optimization formulations, namely the {\em decentralized partially observable Markov decision process (DEC-POMDP)} and { \em Non-cooperative partially observable stochastic game (POSG)}. In DEC-POMDP formulation, we derive a decentralized online learning algorithm to determine the control actions and Lagrangian multipliers (LMs) simultaneously, based on the {\em policy gradient} approach. Under some mild technical conditions, the proposed {\em decentralized policy gradient} algorithm converges almost surely to a local optimal solution. On the other hand, in the non-cooperative POSG formulation, the transmitter nodes are non-cooperative. We extend the decentralized policy gradient solution and establish the
technical proof for almost-sure convergence of the learning algorithms. In both cases, the solutions are very robust to model variations. Finally, the delay performance of the proposed solutions are compared with conventional baseline schemes for interference networks and it is illustrated that substantial delay performance gain and energy savings can be achieved.
\end{abstract}


\newpage

\section{Introduction}
Recently, there have been intense research interests to study the interference channels. In \cite{IA:conventional:2008,IA:M*N_MIMO:2008}, the authors show that interference alignment (using infinite dimension symbol extension in time or frequency selective fading channels) can achieve optimal Degrees-of-freedom (DoF) and the total capacity of the $K$-user interference channels is given by $\frac{K}{2}\log(\text{SNR})+ o(\log(\text{SNR}))$. In \cite{alter:VT:2011,IA:distributed:2008}. the authors consider joint beamforming to minimize the weighted sum MMSE or maximize the SINR of $K$-pairs MIMO interference channels using optimization approaches. In \cite{mimo:game:scutari,mimo:game:Arslan}, the authors considered decentralized beamforming design for MIMO interference networks using non-cooperative games and studied the sufficient conditions for the existence and convergence of the Nash Equilibrium (NE). However, all of these works have assumed that there are infinite backlogs at the transmitters, and focused on the maximization of physical layer throughput. In practice, applications are delay sensitive, and it is critical to optimize the delay performance in the interference network.

%

The design framework taking into consideration of queueing delay and physical layer performance is not trivial as it involves both {\em queuing theory} (to model the queuing dynamics) and {\em information theory} (to model the physical layer dynamics) \cite{Vincent:MIMO}. The simplest approach is to convert the delay constraints into an equivalent average rate constraint using tail probability (large derivation theory), and solve the optimization problem using a purely information theoretical formulation based on the equivalent rate constraint \cite{Hui:07}. However, the control policy thus derived is a function of the channel state information (CSI) only, and it fails to exploit data queue state information (DQSI) in the adaptation process. The Lyapunov drift approach is also widely used in the literature \cite{Neely:2008} to study the queue stability region of different wireless systems and to establish the
throughput optimal control policy (in stability sense). A systematic approach in dealing with delay-optimal resource control in general delay regime is based on the Markov decision process (MDP) technique\cite{Delay_IT:2006,Vincent:MIMO,Cao:2007}. However, brute-force solution of MDP is usually very complex (owing to the curse of dimensionality) and extension to multi-flow problems in interference networks is highly non-trivial.

Another interesting dimension that has been ignored by most of the above works is the inclusion of {\em renewable energy source} on the transmit nodes. For instance, there are intense research interests in exploiting renewable energy in communication network designs\cite{Niyato:2007,Niyato:MC:2007,huang:energy:2011,Kansal:2007}. { In \cite{Niyato:2007,Niyato:MC:2007},
the authors presented an optimal energy management policy for a solar-powered device that uses a sleep
and wake up strategy for energy conservation in wireless sensor networks.} In \cite{huang:energy:2011}, the authors developed
a solar energy prediction algorithm to estimate the amount of energy harvested by solar panels to deploy
power-efficient task management methods on solar energy-harvested wireless sensor nodes. In \cite{Kansal:2007}, the
author proposed a power management scheme under the assumption that the harvested energy satisfies
performance constraints at the application layer. However, in all these works, the delay requirement of applications have been completely ignored. Furthermore, the renewable energy source can act as low cost supplement to the conventional utility power source in communication networks. Yet, there are various technical challenges regarding delay optimal design for interference networks with renewable energy source.

\begin{itemize}

\item {\bf Randomness of Renewable Energy Source:} Recent developments in hardware design have made {\em energy harvesting} possible in wireless communication networks \cite{Sharma:energy:2010,green:Gozalvez}. For example, we have solar-powered base stations available from various telecommunication vendors \cite{green:Gozalvez}. While the renewable energy source may appear to be completely free, there are various challenges involved to fully capture its advantage. For instance, the renewable energy sources are random in nature and energy storage is needed to buffer the unstable supply of renewable energy. Yet, the cost of energy storage depends heavily on the associated capacity \cite{store:2009}. For limited capacity energy storage, the transmission power allocation should be adaptive to the CSI, the DQSI as well as the energy queue state information (EQSI). The CSI, DQSI and EQSI provide information regarding the {\em transmission opportunity}, the {\em urgency of the data flows}, and the {\em available renewable energy}, respectively. It is highly non-trivial to strike a balance among these factors in the optimization.


\item{\bf Decentralized Delay Minimization:} The existing works for the throughput or DoF optimization in the interference network [1-6] requires global knowledge of CSI, which leads to heavy backhaul signaling overhead and high computational complexity for the central controller. For delay minimization with renewable energy source, the entire system state is characterized by the global CSI (CSI from any transmitter to any receiver), the global QSI (data queue length of all users), and the global EQSI (energy queue length of all users). Therefore, the centralized solution (which requires global CSI, DQSI and EQSI) will also induce substantial signaling overhead, which is not practical. It is desirable to have decentralized control based on local observations only. However, due to the partial observation of the system state in decentralized designs, existing solutions of the MDP approach cannot be applied to our problem.

\item{\bf Algorithm Convergence Issue:} In conventional iterative solutions for deterministic network utility maximization (NUM) problems, the updates in the iterative algorithms (such as subgradient search) are performed within the coherence time of the CSI (i.e., the CSI remains quasi-static during the iteration updates) \cite{mimo:game:scutari,mimo:game:Arslan}. When we consider delay minimization, the problem is stochastic and the control actions are defined over ergodic realizations of the system states (CSI, DQSI and EQSI). Furthermore, the restriction of partial observation of system states in decentralized control further complicates the problem. As a result, the convergence proof of the decentralized stochastic algorithm is highly non-trivial.
\end{itemize}

In this paper, we consider delay minimization for interference networks with renewable energy source. The transmitters are capable of harvesting energy from the environment, and the transmission power of a node comes from both the conventional utility power (AC power) and the renewable energy source. For decentralized control, we assume the transmission power of each node is adaptive to the {\em local system states} only, namely the local CSI (LCSI), the local DQSI (LDQSI) and the local EQSI (LEQSI). We consider two delay optimization formulations, namely the {\em decentralized partially observable MDP (DEC-POMDP)}, { which corresponds to a cooperative stochastic game setup (where each user cooperatively share a common system utility),} and { \em non-cooperative partially observable stochastic game (POSG)}, { which corresponds to a non-cooperative stochastic game setup (where each user has a different (and selfish) utility}. In DEC-POMDP formulation, the transmitters are fully cooperative and we derive a decentralized online learning algorithm to determine the control actions and the Lagrangian multipliers (LMs) simultaneously based on the {\em policy gradient} approach \cite{Cao:2007,GPOMDP:experiment:2001}. Under some mild technical conditions, the proposed {\em decentralized policy gradient} algorithm converges almost surely to a local optimal solution. On the other hand, in the non-cooperative POSG formulation, the transmitters are non-cooperative\footnote{Non-cooperative nodes means that each transmitter shall optimize its own utility in a selfish manner.} and we extend the {\em decentralized policy gradient} algorithm and establish the technical proof for almost-sure convergence of the learning algorithms. In both cases, the solutions do not require explicit knowledge of the CSI statistics, random data source statistics as well as the renewable energy statistics. Therefore, the solutions are very robust to model variations. Finally, the delay performance of the proposed solutions are compared with conventional baseline schemes for interference networks and it is illustrated that substantial delay performance gain and energy savings can be achieved by incorporating the CSI, DQSI and EQSI in the power control design.

\section{System Model}\label{sec:model}

We consider $K$-pair interference channels sharing a common spectrum with bandwidth $W$Hz as illustrated in Fig. \ref{fig:system_model}. Specifically, each transmitter maintains a data queue for the random traffic flow towards the desired receiver in the system. { Furthermore, the transmitters are fixed base stations but the receiver can be mobile.} The time dimension is partitioned into scheduling frames (that lasts for $\tau$ seconds). In the following subsections, we shall elaborate the physical layer model, the { random data source model} as well as the renewable energy source model.

\subsection{Physical Layer Model}
The signal received at the $k$-th receiver is given by:
\begin{equation}\label{eq:system_model}
y_k = \underbrace{\sqrt{P_kL_{kk}H_{kk}}x_k}_{\text{desired signal}} + \underbrace{\sum\nolimits_{n\neq k}
\sqrt{P_nL_{kn}H_{kn}}x_n}_{\text{interference}} + z_k,
\end{equation}
where $L_{kn}$ and $H_{kn}$ are the long term path loss and the microscopic channel fading gain respectively, from the $n$-th
transmitter to the $k$-th receiver. $P_k$ is the total transmission power of the $k$-th transmitter. $x_n$ is the information symbol sent by the $n$-th transmitter, and $z_k$ is the additive white Gaussian noise with variance $N_0$. For notation convenience, we define the global CSI as $\mathbf{H}=\{H_{kn},\forall k,n\}$. Furthermore, the assumption on channel model is given as follows.
\begin{assumption}[Channel Model]\label{ass:csi_model}
We assume that the global CSI $\mathbf{H}$ is quasi-static in each frame. Furthermore, $H_{kn}(t)$ is i.i.d. over the scheduling frame
according to a general distribution $\Pr\{H_{kn}\}$ with $\mathbb{E}[H_{kn}] = 1$ and $H_{kn}$
is independent w.r.t. $\{n,k\}$. The path loss $L_{kn}$ remains constant for the duration of the communication session. ~\hfill
\IEEEQED
\end{assumption}

{
Given transmission powers $\{P_k^{tx}\}$, the transmit data rate is given by:
\begin{equation}\label{eq:rate}
R_{k}\leq
W\log_2\left(1+\frac{\xi
P_k^{tx}L_{kk}H_{kk}}{\sum\nolimits_{n\neq
k}
P_nL_{kn}H_{kn}+N_0W}\right),
\end{equation}}
where $\xi\in(0,1]$ is a constant. Note that \eqref{eq:rate} can be used to model both uncoded and coded systems \cite{MQAM:2007}. For example, $\xi=0.5$ for QAM constellation at BER$=1\%$ and $\xi = 1$ for capacity achieving coding (in which \eqref{eq:rate} corresponds to the instantaneous mutual information).


\subsection{{ Random Data Source Model} and Data Queue Dynamics}

Let $\mathbf{A}(t) = \{A_1(t),\cdots,A_K(t)\} $ be the random new arrivals
(number of bits) at the $K$ transmitters at the end of the $t$-th scheduling
frame.
\begin{assumption}[{ Random Data Source Model}]
The arrival process $A_k(t)$ is i.i.d. over the scheduling frame and
is distributed according to a general distribution $\Pr\{A_k\}$ with
average arrival rate $\lambda_{k}=\mathbb{E}[A_k]$. Furthermore, the
random arrival process $\{A_k\}$ is independent w.r.t. $k$.
~\hfill\IEEEQED\label{Def:data}
\end{assumption}

Let $\mathbf{Q}(t) =\{Q_1(t),\cdots,Q_K(t)\}$ denote the
global DQSI in the system, where $Q_k(t)$
represents the number of bits at the queue of transmitter $k$ at the
beginning of frame $t$. { $N_{k}^{Q}$ denotes the maximal buffer size (number of bits) of user $k$.} When
the buffer is full, i.e., $Q_k = N_{k}^{Q}$, new bit arrivals will be
dropped. The cardinality of the global QSI is $I_Q = (1 + N_{k}^{Q})^K$.
Given a new arrival $A_k(t)$ at the end of frame $t$, the
queue dynamics of transmitter $k$ is given by:
\begin{equation}\label{eq:Q_org}\begin{array}{lll}
Q_{k}(t+1) = \Big[\big[Q_k(t)-R_k(t)\tau\big]^+ + A_k(t)
\Big]_{\bigwedge N_{k}^{Q}},
\end{array}
\end{equation}
where $R_k(t)$ is the achievable data rate for receiver $k$ at frame $t$ given in \eqref{eq:rate}, and $[x]_{\bigwedge N_{k}^{Q}}=\min(x,N_{k}^{Q})$.

\subsection{Power Consumption Model with Renewable Energy Source}
The transmission power of each node comes from both the AC power source and the renewable energy source. Specifically, the transmitter is assumed to be capable of harvesting energy from the environment, e.g., using solar panels \cite{green:Gozalvez,green:Jorguseski}. However, the amount of harvestable energy in a frame is random. Let $\mathbf{X}(t) = \{X_1(t),\cdots,X_K(t)\} $ be the harvestable energy (Joule) by the $K$ transmitters during the $t$-th scheduling frame. Note that the harvestable energy $\mathbf{X}(t)$ can be interpreted as the {\em energy arrival} at the $t$-th frame.
\begin{assumption}[Random Renewable Energy Model]
The random process $X_k(t)\geq0$ is i.i.d. over the scheduling frame and is distributed according to a general distribution $\Pr\{X_k\}$ with mean renewable energy $\overline{X}_k = \mathbb{E}[X_k]$. Furthermore, the random process $\{X_k\}$ is independent w.r.t. $k$.

~\hfill\IEEEQED\label{Def:energy}
\end{assumption}

Let $\mathbf{E}(t) =\{E_1(t),\cdots,E_K(t)\}$ denote the global EQSI in the system, where $E_k(t)$ represents the renewable energy level at the energy storage of the $k$-th transmitter at the beginning of frame $t$. { Let $N_{k}^{E}$ denote the maximum energy queue buffer size (i.e., energy storage capacity in Joule) of user $k$.} When the energy buffer is full, i.e., $E_k = N_{k}^{E}$, additional energy cannot be harvested. Given an energy arrival of $X_k(t)$ at the end of frame $t$, the energy queue dynamics of transmitter $k$ is given by:
\begin{equation}\label{eq:EQ_org}\begin{array}{lll}
E_{k}(t+1) = \Big[\big[E_k(t)-P_{k,e}^{tx}(t)\tau\big]^+ + X_k(t)
\Big]_{\bigwedge N_{k}^{E}},
\end{array}
\end{equation}
where $P_{k,e}^{tx}(t)$ is the renewable power consumption that must satisfy the following {\em energy-availability} constraint\footnote{ We consider a discrete time system with fixed time step $\tau$. Hence,  $E_k(t)$ represents the energy level at the renewable energy storage of the $k$-th transmitter at the beginning of frame $t$, and $P_{k,e}^{tx}(t) \tau$ is the renewable energy consumption. As a result, $P_{k,e}^{tx}(t) \tau$ (energy consumed from the renewable energy storage) cannot be larger than $E_k(t)$ (total energy available from the renewable energy storage).}:
\begin{equation}\label{eq:energy_constraint}
P_{k,e}^{tx}(t)\tau \leq E_k(t),\forall k,
\end{equation}

{ The power consumption is contributed by not only the transmission power of the power amplifier (PA) but also the circuit power of the RF chains (such as the mixers, synthesizers and digital-to analog converters). Furthermore, the circuit power $P_{cct}$ is constant irrespective of the transmission data rate. Therefore, the total power consumption of user $k$ at the $t$-th frame is given by
\begin{eqnarray}
P_k(t)=P_{k}^{tx}(t)+P_{cct}\cdot\mathbf{1}\big(P_{k}^{tx}>0\big)
\end{eqnarray}
Note that in practice, due to the random nature of the renewable energy and the limited renewable energy storage capacity, it can be used only as a supplementary form of power rather than completely replacing the AC utility power. To support a total power consumption of $P_k(t)$, we can have power circuitry \cite{Niyato:2007,Niyato:MC:2007} to control the contributions from AC utility $P_{k,ac}(t)$ as well as the renewable energy storage $P_{k,e}(t)$ as illustrated in Fig. \ref{fig:system_model}.  This is similar in concept to hybrid cars where the power is contributed by both the gas engine and the battery. As a result, the total power consumption $P_k(t)$ is given by: $P_k(t) = P_{k,ac}(t) + P_{k,e}(t)$. Given $P_{k,ac}(t)$ and $P_{k,e}(t)$, the transmission power $P_{k}^{tx}(t)$ is given by:
\begin{equation}
P_{k}^{tx}(t) = \Big( P_{k,ac}(t)+P_{k,e}(t)-P_{cct}\cdot\mathbf{1}\big(P_{k}^{tx}>0\big) \Big)^+
\end{equation}}

\section{Delay Optimal Power Control}
\subsection{Control Policy and Resource Constraints}
We define $\boldsymbol{\chi}=\{\mathbf{H},\mathbf{Q},\mathbf{E}\}$ as the global system state, and
$\chi_k=\{\{H_{kn},\forall n\},Q_k,E_k\}$ as the local system state for the $k$-th transmit node, where $\{H_{kn},\forall n\}$ is the LCSI\footnote{We denote the local CSI at the $k$-th transmit node as $\{H_{kn},\forall n\}$. However, in practice, the $k$-th transmit node only needs to observe $H_{kk}$ and the total interference $\sum\nolimits_{n\neq k}P_nL_{kn}H_{kn}$.}, $Q_k$ is the LDQSI and $E_k$ is the LEQSI.
Based on the local system state $\chi_k$, transmitter $k$ determines the power consumption { $\mathbf{P}_k=\{P_{k,ac}\in\mathcal{A}_{ac},P_{k,e}\in\mathcal{A}_{e}\}$} using a control policy defined below, where $\mathcal{A}_{ac}=\{a_{ac}^1,\cdots,a_{ac}^N\}$ and $\mathcal{A}_{e}=\{a_{e}^1,\cdots,a_{e}^N\}$ are the AC power allocation space and the renewable power allocation space (both with cardinality $N$), respectively.
%
\begin{Def}[Stationary Randomized Decentralized Power Control Policy]\label{def:policy} A stationary randomized power control
policy for user $k$, $\Omega_{k}: \chi_k\to \mathcal{P}(\mathcal{A}_{ac},\mathcal{A}_{e})$, is a mapping from the
local system state $\chi_k$ to a probability distribution over the power allocation space $\{\mathcal{A}_{ac},\mathcal{A}_{e}\}$, i.e., $\Omega_{k}(\chi_k)=\mathbf{p}=\{p_{1,1},\cdots,p_{N,N}\}\in\mathcal{P}(\mathcal{A}_{ac},\mathcal{A}_{e})$, where $\mathcal{P}(\mathcal{A}_{ac},\mathcal{A}_{e})=\{\mathbf{p}:\sum_{i,j} p_{i,j}=1 \text{ and } p_{i,j}\geq0,\forall i,j\}$ is the space of joint probability distribution over the power allocations, and $p_{i,j}$ denotes the probability of transmission powers $\{P_{k,ac}=a_{ac}^i,P_{k,e}=a_{e}^j \}$.  ~\hfill\IEEEQED
\end{Def}

For simplicity, denote the joint control policy as $\Omega=\{\Omega_k,\forall k\}$. Note that the power allocation policy $\Omega_k$ should satisfy the energy-availability constraint given in \eqref{eq:energy_constraint}, i.e., given $\chi_k=\{H_{kk},Q_k,E_k\}$, the probability $p_{i,j}$ of transmission powers $\{P_{k,ac}=a_{ac}^i,P_{k,e}=a_{e}^j \}$ satisfy
\begin{equation}
p_{i,j}=0, \text{ if } a_{e}^j > E_k/\tau.
\end{equation}
{ Furthermore, $\Omega_k$ should meet the requirement of circuit power $P_{cct}$ consumption, i.e.,
\begin{equation}
p_{i,j}=0, \text{ if } 0<a_{ac}^i+a_{e}^j < P_{cct}
\end{equation}}
Finally, $\Omega_k$ should also satisfy the per-user average AC power consumption constraint:
\begin{eqnarray}\label{eq:pwr_con}
\overline{P}_k(\Omega)=\lim\sup_{T\rightarrow\infty}\frac{1}{T}\sum_{t=1}^T\mathbb{E}^{\Omega}[P_{k,ac}(t)]\leq
P_k^0,
\end{eqnarray}
where the expectation
in \eqref{eq:pwr_con} is taken w.r.t. the
induced probability measure from the policy
$\Omega$.

{
\begin{Rem}[Formulation with two optimization variables $\{P_{ac},P_{e}\}$]
While the ``reward'' of the system dynamics (the transmission rate in \eqref{eq:rate}) depends on the total transmission power $P$ only, it does not mean the problem can be formulated with just one variable (total transmission power). We also have to look at the ``cost'' side. While the total power consumption $P_{total} = P_{ac} + P_e$, $P_{ac}$ and $P_e$ have different cost structure (and different constraints) as in \eqref{eq:pwr_con} and \eqref{eq:energy_constraint}, respectively. Hence, the problem with $P_{ac}$ and $P_e$ as variables cannot be transformed or reduced into a problem with $P_{total}$ as one variable only (due to the constraints).  ~\hfill\IEEEQED
\end{Rem}
}

\subsection{Parametrization of Control Policy and Dynamics of System State}
In this paper, we consider the parameterized stationary randomized policy, which is widely used in the literature \cite{Marbach:TAC:2001,GPOMDP:experiment:2001,Tao:2001,Buffet:2007}. Specifically, the randomized policy $\Omega_k$ can be parameterized by $\Theta_k$. For example, when a local system state realization $\chi_k$ is observed, { the power consumption of transmit node $k$ is $\mathbf{P}_k=\{P_{k,ac},P_{k,e}\}$ with probability $\mu_{\chi_k}(\Theta_k,\mathbf{P}_k)$ given by\cite{Tao:2001}:
\begin{equation}\label{eq:theta_looktable}
\mu_{\chi_k}(\Theta_k,\mathbf{P}_k)=\left\{
\begin{array}{lll}
\frac{\exp(\theta_{\chi_k,\mathbf{P}_k})}{\sum_{i,j}\exp(\theta_{\chi_k,(a_{ac}^i,a_{e}^j)})\mathbf{1}(a_{e}^j\leq E_k/\tau)} &
\begin{array}{l}
\text{if } P_{k,e} = P_{k,ac} = 0 \text{ or }\\
\text{if } P_{k,e}\leq E_k/\tau
\text{ and } P_{k,e}+P_{k,ac} \geq P_{cct}
\end{array}
\\
0 &\text{ otherwise,}
\end{array}\right.
\end{equation}}
where $\mathbf{1}(\cdot)$ is the indicator function, and $\Theta_k=\{\theta_{\chi_k,\mathbf{P}_k}\in\mathbb{R},\forall \chi_k,\mathbf{P}_k\}$. As a result, the control policy $\Omega_k$ is now parameterized by $\Theta_k$ and is denoted by $\Omega_k^{\Theta_k}$. Another possible parameterization is to use {\em neural network} \cite{Marbach:TAC:2001,GPOMDP:experiment:2001} where the probability is given by:{
\begin{equation}
\mu_{\chi_k}(\Theta_k,\mathbf{P}_k)=\left\{
\begin{array}{lll}
\frac{\exp\big(\theta_{k1}+\sum_{i=2}^{\alpha}\theta_{ki}f_{ki}(\chi_k,\mathbf{P}_k)\big)} {\sum_{i,j}\exp\big(\theta_{k1}+\sum_{i=2}^{\alpha}\theta_{ki}f_{ki}(\chi_k,(a_{ac}^i,a_{e}^j))\big)\mathbf{1}(a_{e}^j\leq E_k/\tau)} & \begin{array}{l}
\text{if } P_{k,e} = P_{k,ac} = 0 \text{ or }\\
\text{if } P_{k,e}\leq E_k/\tau
\text{ and } P_{k,e}+P_{k,ac} \geq P_{cct}
\end{array}
\\
0 &\text{ otherwise,}
\end{array} \right.
\end{equation}}
where $\Theta_k=\{\theta_{ki}\in\mathbb{R}, i=1,\cdots,\alpha\}$ is the parameter and $f_{ki}(\chi_k,\mathbf{P}_k)$ is the prior {\em basis function}. Note that the dimension of the parameter $\Theta_k$ is reduced to $\alpha$ in this case.

For a given stationary parameterized control policy $\Omega^{\Theta}$
($\Theta=\{\Theta_k,\forall k\}$), the induced random
process $\{\boldsymbol{\chi}(t)\}$ is a controlled Markov chain with
transition probability
\begin{equation}\label{eq:sys_tran}
\Pr\{\boldsymbol{\chi}(t+1)|\boldsymbol{\chi}(t),\Omega^{\Theta}\}=\Pr\{H(t+1)\}
\Pr\{\mathbf{Q}(t+1),\mathbf{E}(t+1)|\boldsymbol{\chi}(t),\Omega^{\Theta}(\boldsymbol{\chi}(t))\},
\end{equation}
where the joint data and energy queue transition probability is given by
\begin{equation}\label{eq:Q_tran}
\begin{array}{ll}
&\Pr\{\mathbf{Q}(t+1),\mathbf{E}(t+1)|\boldsymbol{\chi}(t),\Omega^{\Theta}(\boldsymbol{\chi}(t))\}\\
=&\left\{
\begin{array}{lll}
\prod_{k}\Pr\{A_{k}(t)\}\Pr\{X_{k}(t)\}\mu_{\chi_k}(\Theta_k,\mathbf{P}_k(t)) & \text{if } \left. Q_{k}(t+1)=\hat{Q}_k \atop E_k(t+1)=\hat{E}_k \right. \forall k\\
0 & \text{ otherwise,}
\end{array}\right.
\end{array}
\end{equation}
where $\hat{E}_k=\Big[\big[E_k(t)-P_{k,e}(t)\tau\big]^+ + X_k(t)\Big]_{\bigwedge N_{k}^{E}}$, $\hat{Q}_k=\Big[\big[Q_k(t)-R_k(t)\tau\big]^+ +A_k(t)\Big]_{\bigwedge N_{k}^{Q}}$, and $R_k(t)$ is the achievable data rate of receiver $k$ given in \eqref{eq:rate} under the power allocation $\mathbf{P}=\{\mathbf{P}_k(t),\forall k\}$.
{  Note that it is not sufficient to specify the evolution of the joint process $(\chi_1(t),\cdots,\chi_K(t))$ by just describing the measure of individual local processes $\chi_k(t)$. This is because the individual state process $\chi_k(t)$ are not independent and there are mutual coupling.}

Given a unichain policy $\Omega^{\Theta}$, the induced
Markov chain $\{\boldsymbol{\chi}(t)\}$ is ergodic and there exists
a unique steady state distribution $\pi_{\chi}$, where
$\pi_{\chi}(\boldsymbol{\chi})=\lim_{t\to\infty}\Pr\{\boldsymbol{\chi}(t)=\boldsymbol{\chi}\}$\cite{Cao:2007}.
The average delay utility of user $k$, under a unichain policy $\Omega^{\Theta}$, is given by:
\begin{equation}
\label{eq:T_single} \overline{T}_{k}(\Theta)=\lim\sup_{T\rightarrow
\infty}\frac{1}{T}\sum\nolimits_{t=1}^T\mathbb{E}^{\Omega^{\Theta}}[f(Q_{k}(t))],
\end{equation}
where $f(Q_{k})$ is a monotonic increasing utility function of
$Q_{k}$. For example, when $f(Q_{k})=Q_{k}/\lambda_{k}$, using
Little's Law \cite{Cao:2007}, $\overline{T}_{k}(\Omega)$ is the
{\em average delay}\footnote{ Since the buffer size is finite, $\overline{T}_{k}(\Omega)$ is the average delay when $f(Q_k) = Q_k/ ( \lambda_k (1- P_{loss}))$, where $P_{loss}$ is the packet drop rate due to buffer overflow. However in practice our target $P_{loss}\ll 1$, and hence $f(Q_k) = Q_k/ ( \lambda_k )$ is a good approximation for the average delay. Furthermore, this approximation is asymptotically tight as the data buffer size increases. In practice, the approximation error will not be significant since the system will have reasonable $P_{loss}$ (e.g. $0.1\%$).} of user $k$. When $f(Q_{k})=\mathbf{1}(Q_k\geq
Q_k^0)$, $\overline{T}_{k}(\Theta)$ is {\em queue outage probability}\footnote{The probability that the queue
state exceeds a threshold $Q_k^0$, { i.e., $\Pr\{Q_k\geq Q_k^0\}$.}}. { Since $\lambda_k$ is a constant, the average delay $\overline{T}_{k}(\Theta)$ is proportional to the average queue length.}

\subsection{Problem Formulation}\label{sec:prob}
Note that the stochastic dynamics of the $K$ data queues and energy queues are coupled
together via the control policy $\Omega^{\Theta}$. In this paper, we consider two different decentralized control problems:
\subsubsection{DEC-POMDP Problem}
In this case, all the transmitter nodes are cooperative and we seek to
find an optimal stationary control policy $\Omega^{\Theta}$ to minimize a common
weighted sum delay utility in (\ref{eq:T_single}). Since the control policy $\Omega_k^{\Theta_k}$ is only a function of the local system state $\chi_k$, the problem is a partially observed MDP, which is summarized below:
\begin{Prob}[Delay Optimal DEC-POMDP]\label{prob:POMDP}
For some positive constants $\boldsymbol{\beta}=\{\beta_{k},\forall
k\}$, find a stationary control policy $\Omega^{\Theta}$ that minimizes:
\begin{equation}
\label{eq:problem:POMDP}\begin{array}{l}
\min_{\Theta}\overline{T}_{\beta}^{\Theta}=\sum\nolimits_{k}\beta_{k}\overline{T}_{k}(\Theta)=\lim\sup_{T\rightarrow
\infty}\frac{1}{T}\sum\nolimits_{t=1}^T
\mathbb{E}^{\Omega^{\Theta}}\Big[g\big(\boldsymbol{\chi}(t),\Omega^{\Theta}(\boldsymbol{\chi}(t))\big)\Big]\\
\text{subject to}\quad \overline{P}_k(\Omega^{\Theta})=\overline{P}_k(\Theta)\leq
P_k^0, \forall k \\
\quad\quad\quad\quad\quad { E_k\leq N_{k}^E, \forall k}
\end{array},
\end{equation}
where
$g\big(\boldsymbol{\chi}(t),\Omega^{\Theta}(\boldsymbol{\chi}(t))\big)=\sum_{k}\beta_{k}f(Q_k)$
is the joint per-stage utility. The positive constants
$\boldsymbol{\beta}$ indicate the relative importance of the users,
and for the given $\boldsymbol{\beta}$, the solution to
\eqref{eq:problem:POMDP} corresponds to a Pareto optimal point of the
multi-objective optimization problem: $\min_{\Theta}
\overline{T}_{k}(\Theta), \forall k$.  ~\hfill\IEEEQED
\end{Prob}

{ Note that the average AC power constraint is commonly used in a lot of existing studies \cite{Vincent:MIMO,Delay_IT:2006}
and is very relevant in practice (because the electric bill is charged by average AC power consumption $\times$ time of usage). The motivation of Problem \ref{prob:POMDP} is to optimize the delay performance under average cost constraint (AC power) by fully utilizing the free renewable energy. Problem \ref{prob:POMDP} is also equivalent to minimizing the average AC power consumption subject to average delay constraint because they have the same Lagrangian function. }

\subsubsection{Non-Cooperative POSG Problem}
In this case, the $K$ transmitter nodes are non-cooperative and we formulate
the delay utility minimization problem as a non-cooperative partially observable stochastic game (POSG), in which the user $k$ competes
against the others by choosing his power allocation policy $\Omega_{k}^{\Theta_k}$, to maximize
his average utility selfishly. Specifically, the non-cooperative POSG is formulated as Problem \ref{prob:POSG}
\begin{Prob}[Delay Optimal Non-Cooperative POSG]\label{prob:POSG}
For transmitter $k$, we try to find a stationary control policy $\Omega_k^{\Theta_k}$ that minimizes:
\begin{equation}
\label{eq:problem:POSG}\begin{array}{l}
\min_{\Theta_k} \overline{T}_{k}(\Theta_k,\Theta_{-k})=\lim\sup_{T\rightarrow
\infty}\frac{1}{T}\sum\nolimits_{t=1}^T\mathbb{E}^{\Omega_k^{\Theta_k},\Omega_{-k}^{\Theta_{-k}}}[f(Q_{k}(t))]\\
\text{subject to}\quad \overline{P}_k(\Theta_k,\Theta_{-k})\leq
P_k^0,\\
\quad\quad\quad\quad\quad { E_k\leq N_{k}^E, \forall k}
\end{array},\forall k
\end{equation}
where $\Theta_{-k}=\{\Theta_{q=1,q\neq k}^K\}$, and $\Omega_{-k}^{\Theta_{-k}}=\{\Omega_q^{\Theta_q}\}_{q=1,q\neq k}^K$ is the set of all the users' policies except the $k$-th user. ~\hfill\IEEEQED
\end{Prob}

The {\em local equilibrium} solutions of the non-cooperative POSG \eqref{eq:problem:POSG} are formally defined as follows.
\begin{Def}[Local Equilibrium of Non-Cooperative POSG] A profile of the power allocation policy $\Omega^{\Theta^*}=\{\Omega_1^{\Theta_1^*},\cdots,\Omega_K^{\Theta_K^*}\}$ is the local equilibrium of the game \eqref{eq:problem:POSG} if it satisfies the following fixed point equations for some $\gamma^*=\{\gamma_k^*\geq0,\forall k\}$,
\begin{equation}\label{eq:LE}
\begin{array}{lll}
&\nabla_{\Theta_k}\psi_k(\Theta_k^{*},\Theta_{-k}^{*},\gamma_k^*)=\mathbf{0}, & \nabla_{\Theta_k\Theta_k}^2\psi_k(\Theta_k^{*},\Theta_{-k}^{*},\gamma_k^*)\succ \mathbf{0}\nonumber \\
\text{and} & \overline{P}_k(\Theta_k^*,\Theta_{-k}^*)-P_k^0\leq0,& \gamma_k^*\left(\overline{P}_k(\Theta_k^*,\Theta_{-k}^*)-P_k^0\right)=0
\end{array}\quad \forall k,
\end{equation}
{ where $\psi_k(\Theta_k,\Theta_{-k},\gamma_k)=\overline{T}_k(\Theta_k,\Theta_{-k})+ \gamma_k(\overline{P}_k(\Theta_k,\Theta_{-k})-P_k^0)$. ~\hfill\IEEEQED}
\end{Def}
\begin{Rem}[Interpretation of the Local Equilibrium]
$\psi_k(\Theta_k,\Theta_{-k},\gamma_k)$ can be regarded as the {\em Lagrange function} for user $k$ (given the policies of the other users $\Theta_{-k}$) in the non-cooperative POSG problem  \eqref{eq:problem:POSG}. From the Lagrangian theory\cite{ber:book:1999}, a local equilibrium $\Omega^{\Theta^*}=\{\Omega_1^{\Theta_1^*},\cdots,\Omega_K^{\Theta_K^*}\}$ means that given $\Omega_{-k}^{\Theta_{-k}^*}$, $\Omega_{k}^{\Theta_{k}^*}$ is the local optimal solution for the non-cooperative POSG problem \eqref{eq:problem:POSG}. ~\hfill\IEEEQED
\end{Rem}

\begin{Rem}[Comparison between the DEC-POMDP and { Non-Cooperative} POSG Problems]
In Problem \ref{prob:POMDP} (DEC-POMDP), the controller is decentralized at the $K$ transmitters and they have access to the local system state only. Yet, the $K$ controllers are fully cooperative in the sense that they are designed to optimize a common objective function where the per-stage utility is assumed to be known globally through message passing. As a result, they interact in a {\em decentralized cooperative manner}. On the other hand, in the non-cooperative POSG formulation, the $K$ controllers are {\em non-cooperative} in the sense that each controller is interested in optimizing its own delay utility function. Hence, they interact in a {\em decentralized non-cooperative manner}. ~\hfill\IEEEQED
\end{Rem}

Note that the policies $\{\Omega_k,\forall k\}$ are {\em reactive} or {\em memoryless} in that their choice of action is based only upon the current local observation. Furthermore, the DEC-POMDP and the non-cooperative POSG problem are NP-hard \cite{Meuleau:1999}. Instead of targeting at global optimal solutions, we shall derive low complexity iterative algorithms for local optimal solutions in the following sections.


\section{Decentralized Solution for DEC-POMDP}\label{sec:opt_solution}

In this section, we shall propose a decentralized online policy gradient update algorithm to find a local optimal solution for problem \eqref{eq:problem:POMDP}. The proposed solution has low complexity and does not require explicit knowledge of the CSI statistics, random data source statistics as well as the renewable energy statistics.

\subsection{Decentralized Stochastic Policy Gradient Update}
We first define the {\em Lagrangian function} of problem \eqref{eq:problem:POMDP} as
\begin{equation}
\psi(\Theta,\boldsymbol{\gamma})=\sum_k\left(\beta_k\overline{T}_k(\Theta)+ \gamma_k(\overline{P}_k(\Theta)-P_k^0) \right),
\end{equation}
where $\boldsymbol{\gamma}=\{\gamma_k\in\mathbb{R}^+,\forall k\}$ is the LM vector w.r.t. the average power constraint for all the users. The local optimal solution $\Theta^*$ for problem \eqref{eq:problem:POMDP} should satisfy the following first-order necessary conditions given by \cite{ber:book:1999}
\begin{equation}\label{eq:kkt}
\begin{array}{lll}
\nabla_{\Theta}\psi(\Theta^*,\boldsymbol{\gamma}^*) &=& \mathbf{0}\\
\gamma_k^*(\overline{P}_k(\Theta^*)-P_k^0) &=& 0, \forall k
\end{array}
\end{equation}

Define a reference state\footnote{For example, we can set $\{Q_k^I=N_{Q},E_k^I=N_{E},\forall k\}$ without loss of optimality.} $\{\mathbf{Q}^I,\mathbf{E}^I\}=\{\{Q_1^I,\cdots,Q_K^I\},\{E_1^I,\cdots,E_K^I\}\}$ and using perturbation analysis \cite{Cao:2007,Marbach:TAC:2001}, the gradient\footnote{Note that a change of $\Theta$ will affect the function $\psi(\Theta,\boldsymbol{\gamma})$ via the probability measure behind the expectation in $\psi(\Theta,\boldsymbol{\gamma})$ and hence, deriving the gradient is highly non-trivial. } $\nabla_{\Theta}\psi(\Theta,\boldsymbol{\gamma})$ is given in the following lemma.
\begin{Lem}[Gradient of the Lagrangian Function]\label{lem:gradient}
The gradient of the Lagrangian function is given by
\begin{equation}
\begin{array}{lll}
\nabla_{\Theta_k}\psi(\Theta,\boldsymbol{\gamma})=\sum_{\boldsymbol{\chi}}\sum_{\mathbf{P}} \pi(\boldsymbol{\chi};\Theta)\mu_{\boldsymbol{\chi}}(\Theta,\mathbf{P}) \frac{\nabla_{\Theta_k}\mu_{\chi_k}(\Theta_k,\mathbf{P}_k)}{\mu_{\chi_k}(\Theta_k,\mathbf{P}_k)}
q(\boldsymbol{\chi},\mathbf{P};\boldsymbol{\gamma},\Theta)
\end{array}
\end{equation}
where $\pi(\boldsymbol{\chi};\Theta)$ is the steady state probability of state $\boldsymbol{\chi}$ under the policy $\Omega^{\Theta}$, $\mu_{\boldsymbol{\chi}}(\Theta,\mathbf{P})=\prod_k\mu_{\chi_k}(\Theta_k,\mathbf{P}_k)$ is the probability that joint action $\mathbf{P}$ is taken, and $\frac{\nabla_{\Theta_k}\mu_{\chi_k}(\Theta_k,\mathbf{P}_k)}{\mu_{\chi_k}(\Theta_k,\mathbf{P}_k)}\triangleq\mathbf{0}$, if $\mu_{\chi_k}(\Theta_k,\mathbf{P}_k)=0$,
\begin{equation}
q(\boldsymbol{\chi},\mathbf{P};\boldsymbol{\gamma},\Theta)
=\mathbb{E}^{\Omega^{\Theta}}\Big[ \sum_{t=0}^{T^I-1}\big(g_\psi(\boldsymbol{\chi}(t),\mathbf{P}(t))-\psi(\Theta,\boldsymbol{\gamma})\big)|\boldsymbol{\chi}(0)=\boldsymbol{\chi},\mathbf{P}(0)=\mathbf{P}
\Big],
\end{equation}
where $g_\psi(\boldsymbol{\chi},\mathbf{P})=\sum_k \beta_k f(Q_k)+\gamma_k(P_{k,ac}-P_k^0)$. $T^I=\min\{t>0|\mathbf{Q}(t)=\mathbf{Q}^I,\mathbf{E}(t)=\mathbf{E}^I\}$ is the first future time that the reference state $\{\mathbf{Q}^I,\mathbf{E}^I\}$ is visited. ~\hfill\IEEEQED
\end{Lem}
\begin{proof}
Please refer to Appendix \ref{app:gradient}.
\end{proof}

Note that the brute force solution of \eqref{eq:kkt} requires huge complexity and knowledge of the CSI statistics, random data source statistics as well as the renewable energy statistics. Based on Lemma \ref{lem:gradient}, we shall propose a low complexity decentralized online policy gradient update algorithm to obtain a solution of \eqref{eq:kkt}. Specifically, the key steps for decentralized online learning is given below.
\begin{itemize}
\item{\bf Step 1, Initialization}: Each transmitter initiates the local parameter $\Theta_k$.
\item{\bf Step 2, Per-user Power Allocation}: At the beginning of the
$t$-th frame, each transmitter determines the transmission power allocation according to the policy $\Omega_k^{\Theta_k}$ based on the local system state $\chi_k$, and transmit at the associated achievable data rate given in \eqref{eq:rate}.
\item{\bf Step 3, Message Passing among the $K$ Transmitters}\footnote{ Note that the per-user per-stage utility includes not only the packet buffer states but also the control action. As a result, just broadcasting nodes' buffer states is not enough to replace the per-user per-stage utility. Furthermore, if each user wants to have complete state information, they need to share both the buffer states and the CSI states. As a result, it will cause much information exchanges compared with the per-user per-stage utility sharing. Table \ref{tab:comp} summarizes the communication overhead by exchanging the per-stage utility and sharing the buffer states and the CSI states.}: At the end of the
$t$-th frame, each transmitter shares the per-user per-stage utility $g_{L,k}= \beta_k f(Q_k)+\gamma_k(P_{k,ac}-P_k^0)$ and the reference state indication $\zeta_k$, where $\zeta_k=1$ if $\{Q_k=Q_k^I,E_k=E_k^I\}$, and $\zeta_k=0$ otherwise. 
\item{\bf Step 4, Per-user Parameter $\Theta_k$ Update}:
Based on the current local observation, each of the transmitters updates the
local parameter $\Theta_k$ according to Algorithm
\ref{alg:learning}.
\item{\bf Step 5, Per-user LM Update:}
Based on the current local observation, each of the transmitters updates the
local LMs $\{\gamma_k,\forall k\}$ according to Algorithm
\ref{alg:learning}.
\end{itemize}

Fig. \ref{fig:learning_pomdp} illustrates the above procedure by a flowchart. The detailed algorithm for the local parameters and LMs update in Step 4 and Step 5 is given below:
\begin{Alg}[Online Learning Algorithm for Per-user Parameter and LM]
\label{alg:learning}
Let $\chi_k=\{H_{kk},Q_k,E_k\}$ be the current local system state, $\mathbf{P}_k$ be the current realization
of power allocation, $g_L=\sum_k g_{L,k}$ be the current realization of the per-stage utility and $\zeta=\prod_k\zeta_k$ be the current realization of the reference state indication. The online learning algorithm at the $k$-th transmitter is given by
\begin{equation}
\begin{array}{ll}\label{eq:learning}
\Theta_k^{t+1} &= \Theta_k^{t} -a(t)\left(g_L -\widetilde{L}^t\right)z_k^t   \\
\gamma_k^{t+1}  &=\left[ \gamma_k^{t}+b(t)\left( P_{k,ac} - P_k^0   \right)\right]^+,
\end{array}
\end{equation}
where $\widetilde{L}^{t+1} =\widetilde{L}^{t+1}-a(t)\left( g_L - \widetilde{L}^t  \right)$, and
\begin{equation}\label{eq:zt_pomdp}
z_k^{t+1}=\left\{\begin{array}{lll}
\frac{\nabla_{\Theta_k}\mu_{\chi_k}(\Theta_k^t,\mathbf{P}_k)}{\mu_{\chi_k}(\Theta_k^t,\mathbf{P}_k)} & \text{if } \zeta=1\\
z_k^{t}+\frac{\nabla_{\Theta_k}\mu_{\chi_k}(\Theta_k^t,\mathbf{P}_k)}{\mu_{\chi_k}(\Theta_k^t,\mathbf{P}_k)} & \text{otherwise.}
\end{array}
\right.
\end{equation}
Stepsizes $\{a(t),b(t)\}$ are non-increasing positive scalars satisfying $\sum_{t}a(t)=\sum_{t}b(t)=\infty$, $\sum_{t}(a(t)^2+b(t)^2)<\infty$, $\frac{b(t)}{a(t)}\to 0$. ~\hfill\IEEEQED

\end{Alg}
\begin{Rem}[Feature of the Learning Algorithm \ref{alg:learning}] The learning algorithm only requires local observations only, i.e., local system state $\{H_{kk},Q_k,E_k\}$ at each transmit node, and limited message passing of $\{\zeta_k,g_{L,k}\}$, where the overhead is quite mild\cite{Palomar:distributed:2008}. Both the per-user parameter and the LMs are updated
simultaneously and distributively at each transmitter. { Furthermore, the iteration is online and proceed in the same timescale as the CSI and QSI variations in the learning algorithm. Finally, the solution does not require knowledge
of the CSI distribution or statistics of the arrival process or renewable energy process, i.e., robust to model variations.} ~\hfill\IEEEQED
\end{Rem}

\subsection{Convergence Analysis}\label{sec:pomdp_convergence}
In this section, we shall establish the convergence proof of the
proposed decentralized learning algorithm \ref{alg:learning}. { Since we
have two different stepsize sequences $\{a(t)\}$ and
$\{b(t)\}$ with $b(t) =
o(a(t))$, e.g., $a(t) = \frac{1}{t^{2/3}}$ and $b(t) = \frac{1}{n}$. the per-user parameter updates and the LM
updates are done simultaneously but over two different timescales.
During the per-user parameter update (timescale I), we have
$\gamma_{k}^{t+1}-\gamma_{k}^{t}=O(b(t))=
o(a(t)), \forall k$. Therefore, the LMs appear to be
quasi-static\cite{Borkar:2008} during the per-user parameter update
in \eqref{eq:learning}, and the convergence analysis can be
established over two timescales separately.} We first have the
following lemma.
\begin{Lem}[Convergence of Per-user Parameter Learning (Timescale I)]\label{lem:potential_converge}
The iterations of the per-user parameter
$\Theta^t$ in the proposed learning algorithm
\ref{alg:learning} will converge almost surely to a stationary point,
i.e.,
$\lim_{t\rightarrow\infty}\Theta^t=\Theta^{\infty}(\boldsymbol{\gamma})$, and $\Theta^{\infty}(\boldsymbol{\gamma})$ satisfies
\begin{equation}
\nabla_{\Theta}\psi(\Theta^{\infty}(\boldsymbol{\gamma}),\boldsymbol{\gamma})=\mathbf{0}. 
\end{equation}
\end{Lem}
\begin{proof}
Please refer to Appendix \ref{app:potential_converge}.
\end{proof}

On the other hand,
during the LM update (timescale II), we have
$\lim_{t\to\infty}|\Theta^t-\Theta^{\infty}(\boldsymbol{\gamma}^t)|=0$
almost surely. Hence, during the LM update in
\eqref{eq:learning}, the per-user parameter is seen as almost
equilibrated. The convergence of the LMs is summarized below.

\begin{Lem}[Convergence of LM over Timescale II]\label{lem:lm_converge} The iterations of the LMs $\lim_{t\to \infty} \boldsymbol
\gamma^t=\boldsymbol \gamma^{\infty}$ almost surely, where $\boldsymbol
\gamma^{\infty}$ satisfies the power constraints of all the users in
\eqref{eq:pwr_con}. 
~ \hfill\QED
\end{Lem}
\begin{proof}
Please refer to Appendix \ref{app:lm_converge}.
\end{proof}

Based on the above lemmas, we can summarize the convergence
performance of the proposed
learning algorithm in the following theorem.
\begin{Thm}[Convergence of Online Learning Algorithm \ref{alg:learning}]
In the learning algorithm \ref{alg:learning}, we have $(\Theta^{t},\boldsymbol{\gamma}^t)\to(\Theta^{\infty},\boldsymbol{\gamma}^{\infty})$ almost surely, where $\Theta^{\infty}$ and $\boldsymbol{\gamma}^{\infty}$ satisfy the KKT condition given in \eqref{eq:kkt}, i.e.,
\begin{equation}
\nabla_{\Theta}\psi(\Theta^{\infty},\gamma^{\infty})=\mathbf{0}, \gamma_k^{\infty}(\overline{P}_k(\Theta^{\infty})-P_k^0) = 0 
\end{equation}
and the power constraints of all the users in \eqref{eq:pwr_con}. Furthermore, if $\nabla_{\Theta\Theta}^2\psi(\Theta^{\infty},\gamma^{\infty})\succ\mathbf{0}$ (positive definite matrix), then $\Theta^{\infty}$ is a local optimal solution for the constrained DEC-POMDP problem in \eqref{eq:problem:POMDP}. ~ \hfill\IEEEQED
\end{Thm}

Note that $\nabla_{\Theta\Theta}^2\psi(\Theta^{\infty},\gamma^{\infty})\succ\mathbf{0}$ is a very mild condition that is usually satisfied \cite{Borkar:2008}.


\section{Decentralized Solution for { Non-Cooperative} POSG Problem}
In this section, we shall propose a decentralized online policy gradient update algorithm to find a
local equilibrium of the { non-cooperative} POSG problem. The proposed solution also has low complexity and does not require explicit knowledge of the CSI statistics, random data source statistics as well as the renewable energy
statistics.

\subsection{Decentralized Stochastic Policy Gradient Update}

From \eqref{eq:LE}, the Lagrangian function for user $k$ is given by
\begin{equation}\label{eq:lm_fun_posg}
\psi_k(\Theta_k,\Theta_{-k},\boldsymbol{\gamma})=\beta_k\overline{T}_k(\Theta_k,\Theta_{-k})+ \gamma_k(\overline{P}_k(\Theta_k,\Theta_{-k})-P_k^0),
\end{equation}
where $\gamma_k\in\mathbb{R}^+$ is the LM w.r.t. the average power constraint for user $k$.
Following similar perturbation analysis as in Lemma \ref{lem:gradient}, the gradient $\nabla_{\Theta_k}\psi_k(\Theta_k,\Theta_{-k},\gamma_k)$ is given in the following lemma.
\begin{Lem}[Gradient of the Lagrangian Function]\label{lem:gradient_posg}
The gradient of the Lagrangian function in \eqref{eq:lm_fun_posg} is given by
\begin{equation}
\begin{array}{lll}
\nabla_{\Theta_k}\psi_k(\Theta_k,\Theta_{-k},\gamma_k)=\sum_{\boldsymbol{\chi}}\sum_{\mathbf{P}} \pi(\boldsymbol{\chi};\Theta)\mu_{\boldsymbol{\chi}}(\Theta,\mathbf{P}) \frac{\nabla_{\Theta_k}\mu_{\chi_k}(\Theta_k,\mathbf{P}_k)}{\mu_{\chi_k}(\Theta_k,\mathbf{P}_k)}
q_k(\boldsymbol{\chi},\mathbf{P};\gamma_k,\Theta),
\end{array}
\end{equation}
where
\begin{equation}
q_k(\boldsymbol{\chi},\mathbf{P};\gamma_k,\Theta)
=\mathbb{E}^{\Omega^{\Theta}}\Big[ \sum_{t=0}^{T^I-1}\big(f(Q_k(t))+\gamma_k(P_{k,ac}(t)-P_k^0)-\psi_k(\Theta_k,\Theta_{-k},\gamma_k)\big)|\boldsymbol{\chi}(0)=\boldsymbol{\chi},\mathbf{P}(0)=\mathbf{P}
\Big].
\end{equation}
 ~\hfill\IEEEQED
\end{Lem}

Based on the Lemma \ref{lem:gradient_posg}, we shall propose a low complexity decentralized online policy gradient update algorithm to obtain a local equilibrium. Specifically, the key steps for decentralized online learning is given below.
\begin{itemize}
\item{\bf Step 1, Initialization}: Each transmitter initiates the local parameter $\Theta_k$.
\item{\bf Step 2, Per-user Power Allocation}: At the beginning of the
$t$-th frame, each transmitter determines the transmission power allocation according to the policy $\Omega_k^{\Theta_k}$ based on the local system state $\chi_k$, and transmit at the associated achievable data rate given in \eqref{eq:rate}.
\item{\bf Step 3, Message Passing among the $K$ Transmitters}: At the end of the
$t$-th frame, each transmitter shares the one bit reference state indication $\zeta_k$, where $\zeta_k=1$ if $\{Q_k=Q_k^I,E_k=E_k^I\}$, and $\zeta_k=0$ otherwise.
\item{\bf Step 4, Per-user Parameter $\Theta_k$ Update}:
Based on the current local observation, each of the transmitters updates the
local parameter $\Theta_k$ according to Algorithm \ref{alg:learning_posg}.
\item{\bf Step 5, Per-user LM Update:}
Based on the current local observation, each of the transmitters updates the
local LMs $\{\gamma_k,\forall k\}$ according to Algorithm \ref{alg:learning_posg}.
\end{itemize}

Fig. \ref{fig:learning_posg} illustrates the above procedure by a flowchart. The detailed algorithm for the local parameters and LMs update in Step 4 and Step 5 is given below:
\begin{Alg}[Online Learning Algorithm for Per-user Parameter and LM]
\label{alg:learning_posg}
Let $\chi_k=\{H_{kk},Q_k,E_k\}$ be the current local system state, $\mathbf{P}_k$ be the current realization
of power allocation, $\zeta=\prod_k\zeta_k$ be the current realization of the reference state indication. The online learning algorithm at the $k$-th transmitter is given by
\begin{equation}
\begin{array}{ll}\label{eq:learning_posg}
\Theta_k^{t+1} &= \Theta_k^{t} -a(t)\left( f_k(Q_k)+\gamma_k^t(P_{k,ac}-P_k^0) -\widetilde{L}_k^t\right)z_k^t   \\
\gamma_k^{t+1}  &=\left[ \gamma_k^{t}+b(t)\left( P_{k,ac} - P_k^0   \right)\right]^+,
\end{array}
\end{equation}
where $\widetilde{L}_k^{t+1} =\widetilde{L}_k^{t}-a(t)\left( f_k(Q_k)+\gamma_k^t(P_{k,ac}-P_k^0) - \widetilde{L}_k^t  \right)$, and
\begin{equation}
z_k^{t+1}=\left\{\begin{array}{lll}
\frac{\nabla_{\Theta_k}\mu_{\chi_k}(\Theta_k^t,\mathbf{P}_k)}{\mu_{\chi_k}(\Theta_k^t,\mathbf{P}_k)} & \text{if } \zeta=1\\
z_k^t+\frac{\nabla_{\Theta_k}\mu_{\chi_k}(\Theta_k^t,\mathbf{P}_k)}{\mu_{\chi_k}(\Theta_k^t,\mathbf{P}_k)} & \text{otherwise.}
\end{array}
\right.
\end{equation}

 ~\hfill\IEEEQED

\end{Alg}
\begin{Rem}[Features of the Learning Algorithm \ref{alg:learning_posg}] The learning algorithm only requires local observations, i.e., local system state $\{H_{kk},Q_k,E_k\}$ at each transmit node, and one bit message passing of $\zeta_k$. Both the per-user parameter and the LMs are updated simultaneously and distributively at each transmitter. { Furthermore, the iteration is online and proceed in the same timescale as the CSI and QSI variations in the learning algorithm. Finally, the solution does not require knowledge
of the CSI distribution or statistics of the arrival process or renewable energy process, i.e., robust to model variations.}   ~\hfill\IEEEQED
\end{Rem}

\subsection{Convergence Analysis}
In this section, we shall establish the convergence proof of the
proposed decentralized learning algorithm \ref{alg:learning_posg}. Specifically, let $\eta=\max_{k,n\neq k}\frac{L_{kn}}{L_{kk}}$, and let $\mathcal{F}^*=\{\Theta^*\}$ be the set of the local equilibrium of the game \eqref{eq:problem:POSG}, i.e., $\Theta^*$ satisfies the fixed point equations in \eqref{eq:LE}. The convergence
performance of the proposed learning algorithm is given in the following theorem.
\begin{Thm}[Convergence of Online Learning Algorithm \ref{alg:learning_posg}]\label{thm:posg}
Suppose $\mathcal{F}^*$ is not empty. The iterations of the per-user parameter
$\Theta^t$ in the proposed learning algorithm \ref{alg:learning_posg} will converge almost surely to an invariant set given by
\begin{equation}\label{eq:set_posg}
\mathcal{S}_{\theta}\triangleq\{ \Theta: ||\Theta - \Theta^*||-\delta\leq0\}
\end{equation}
as $t\to\infty$, for some positive constant $\delta=O(\eta^2)$ and some $\Theta^*\in\mathcal{F}^*$. ~\hfill\IEEEQED
\end{Thm}
\begin{proof}
Please refer to Appendix \ref{app:thm:posg}
\end{proof}
\begin{Rem}[Interpretation of Theorem \ref{thm:posg}]
From \eqref{eq:set_posg}, the error between the converged solution $\Theta^{\infty}$ and the local equilibrium of the POSG $\Theta^*$ decreases in the order of $\eta^2$ where $\eta$ represents the degree of coupling among the transmitters. ~\hfill\IEEEQED
\end{Rem}

\section{Simulations}
In this section, we shall compare the performances of the proposed decentralized solutions against various existing decentralized
baseline schemes.
\begin{itemize}
\item{\bf Baseline 1, Orthogonal Transmission:} The transmissions between the $K$ pairs are coordinated using TDMA so that there is no interference among the users. Both the AC and renewable power consumption are adaptive to LCSI and LEQSI only by optimizing the sum throughput as in \cite{Sharma:energy:2010}.
\item{\bf Baseline 2, LCSI and LEQSI Only Strategy:} The $K$ transmitters send data to their desired receiver simultaneously sharing the same spectrum. Both the AC and renewable power consumption are adaptive to LCSI and LEQSI only by optimizing the sum throughput as in \cite{Sharma:energy:2010}.
\item{\bf Baseline 3, Greedy Strategy:} The $K$ transmitters send data to their desired receiver simultaneously sharing the same spectrum. The transmitters will consume all the available renewable energy source at each frame (emptying the renewable energy buffer at each frame), and the AC power consumption is adaptive to LCSI only by optimizing the sum throughput.
\end{itemize}

In the simulation, we consider a symmetric system where $\frac{L_{ki}}{L_{kk}}=0.1,\forall k,n\neq k$ as in \cite{mimo:game:Arslan}.
The long term path loss for the desired link is 15dB, which corresponds to a cell size of 5.6km\cite{ITU:1997}. The static circuit power is $P_{cct} = 40$ (Watt) \cite{Arnold:2010}. We assume Poisson packet arrival\footnote{ Note that the proposed algorithm works for generic packet and renewable energy arrival models as depicted in Definition \ref{Def:data} and Definition \ref{Def:energy}. The Poisson model is used for simulation illustration only.}  with average arrival rate $\lambda_k$ (packet/s) and exponentially distributed random packet size with mean $\overline{N}_k$ = 2Mbits. The scheduling frame duration $\tau$ is 50ms, and the total BW is $W$ = 1MHz. The maximum data queue buffer size $N_{k}^{Q}$ is 5 (packets). Furthermore, we consider Poisson energy arrival with average arrival rate $\overline{X}_k$ (Watt) as in \cite{Sharma:energy:2010}, and the renewable energy is stored in a 1.2V 20Ah lithium-ion battery. The AC power allocation space and the renewable power allocation space is given by $\mathcal{A}_{ac}=\mathcal{A}_e=[0,300,600,900,1200,1500]$ (Watt). The average delay is considered as our utility ($f(Q_k)=Q_k/\lambda_k$), and the randomized policy $\Omega_k$ is parameterized in the form given by \eqref{eq:theta_looktable}.

\subsection{Delay Performance w.r.t. the AC power consumption}
Fig. \ref{fig:acpower} illustrates the average delay per user versus the AC power consumption $P_k^0$. The average data arrival rate is $\lambda_k=1.1$, and the energy arrival rate is $\overline{X}_k=800$. The average delay of all the schemes decreases as the AC power consumption increase, and the proposed schemes achieve significant performance gain over all the baselines. This gain is contributed by the DQSI and EQSI aware dynamic power allocation. Furthermore, it can also be observed that the solution to the non-cooperative POSG problem has similar performance as the solution to the DEC-POMDP problem.

\subsection{Delay Performance w.r.t. Number of Power Control Levels}
Fig. \ref{fig:rate} illustrates the average delay per user versus the number of power control levels that lie between 0 and 1.5kW. The average data arrival rate is $\lambda_k=1.1$, the energy arrival rate is $\overline{X}_k=800$, and the average AC power consumption is $P_k^0=800$. The average delay of the proposed schemes decreases as the number of power control levels increases, yet the performance improvement is marginal. It can also be observed that there is significant performance gain with the proposed schemes compared with all the baselines, and the solution to the non-cooperative POSG problem has similar performance as the solution to the DEC-POMDP problem.

\subsection{Delay Performance w.r.t. Renewable Energy Buffer Size}
Fig. \ref{fig:buffer} illustrates the average delay per user versus the renewable energy buffer size $N_{k}^{E}$. Specifically, we consider the lithium-ion battery given from 1.2V 10Ah to 40Ah. The average data arrival rate is $\lambda_k=1.1$, the energy arrival rate is $\overline{X}_k=800$, and the average AC power consumption is $P_k^0=500$. It can also be observed that the proposed schemes achieve significant performance gain over all the baselines at any given renewable energy buffer size.

\subsection{Convergence Performance}
Fig. \ref{fig:convergence} illustrates the convergence property of the proposed schemes. We plot the randomized power control policy $\mu_{\chi_1}(\Theta_1,\mathbf{P}_1)$ versus the scheduling frame index for the POMDP and non-cooperative POSG problems, respectively. The average data arrival rate is $\lambda_k=1.1$, the energy arrival rate is $\overline{X}_k=800$, and the average AC power consumption is $P_k^0=1100$. It can be observed that the convergence rate of the online algorithm is quite fast. For example, the delay
performance of the proposed schemes already out-performs all the baselines at the 2500-th scheduling frame. Furthermore, the delay performance at the 2500-th scheduling frame is already quite close to the converged average delay.

\black

\section{Conclusion}

In this paper, we consider the decentralized delay minimization for interference networks with limited renewable energy storage. Specifically, the transmitters are capable of harvesting energy from the environment, and the transmission power of a node comes from both the conventional utility power (AC power) and the renewable energy source. We consider two decentralized delay optimization formulations, namely the DEC-POMDP and the non-cooperative POSG, where the control policy is adaptive to local system states (LCSI, LDQSI and LEQSI) only. In the DEC-POMDP formulation, the controllers interact in a cooperative manner and the proposed decentralized policy gradient solution converges almost surely to a local optimal point under some mild technical conditions. In the non-cooperative POSG formulation, the transmitter nodes are non-cooperative. We extend the decentralized policy gradient solution and establish the technical proof for almost-sure convergence of the learning algorithms. In both cases, the solutions are very robust to model variations. Finally, the delay performance of the proposed solutions are compared with conventional baseline schemes for interference networks and it is illustrated that substantial delay performance gain and energy savings can be achieved by incorporating the CSI, DQSI and EQSI in the power control design.

\appendices
\section{Proof of Lemma \ref{lem:gradient}}\label{app:gradient}

From the perturbation analysis \cite{Cao:2007,Marbach:TAC:2001} in MDP, the gradient $\nabla_{\Theta}\psi(\Theta)$\footnote{The notation of $\boldsymbol{\gamma}$ is ignored in this section for simplicity.} is given by
\begin{equation}\label{eq:gradient_org}
\nabla_{\Theta} \psi(\Theta) = \sum_{\boldsymbol{\chi}}\pi(\boldsymbol{\chi};\Theta)\{\nabla_{\Theta} g_\psi( \boldsymbol{\chi},\Theta) +
\sum_{ \boldsymbol{\chi}^{'} } (\nabla_{\Theta} \Pr\{ \boldsymbol{\chi}^{'}| \boldsymbol{\chi},\Theta\})V( \boldsymbol{\chi}^{'} ) \},
\end{equation}
where $V( \boldsymbol{\chi} )$ satisfies the following Bellman (Possion) equation
\begin{equation}\label{eq:bellV}
V(\boldsymbol{\chi})+\psi(\Theta) = g_\psi( \boldsymbol{\chi},\Theta) + \sum_{\boldsymbol{\chi}^{'}} \Pr\{ \boldsymbol{\chi}^{'}| \boldsymbol{\chi},\Theta\}V( \boldsymbol{\chi}^{'}).
\end{equation}

Since $\sum_{\mathbf{P}}\mu_{\boldsymbol{\chi}}(\Theta,\mathbf{P})=1$ for every $\Theta$, we have
\begin{equation}\label{eq:gpr}
\begin{array}{lll}
\nabla_{\Theta_k} g_\psi( \boldsymbol{\chi},\Theta) &=& \sum_{\mathbf{P}}\mu_{\boldsymbol{\chi}}(\Theta,\mathbf{P})
\frac{\nabla \mu_{\chi_{k}}(\Theta_k,\mathbf{P})}{\mu_{\chi_{k}}(\Theta_k,\mathbf{P})} ( g_\psi( \boldsymbol{\chi},\mathbf{P}) - \psi(\Theta) )\\
\nabla_{\Theta_k} \Pr\{ \boldsymbol{\chi}^{'}| \boldsymbol{\chi},\Theta\}&=&\sum_{\mathbf{P}}\mu_{\boldsymbol{\chi}}(\Theta,\mathbf{P})
\frac{\nabla \mu_{\chi_{k}}(\Theta_k,\mathbf{P})}{\mu_{\chi_{k}}(\Theta_k,\mathbf{P})}
\Pr\{\boldsymbol{\chi}^{'}| \boldsymbol{\chi},\mathbf{P}\}.
\end{array}
\end{equation}

Substituting \eqref{eq:gpr} into \eqref{eq:gradient_org}, we have
\begin{equation}
\nabla_{\Theta_k} \psi(\Theta)=\sum_{\boldsymbol{\chi}}\sum_{\mathbf{P}} \pi(\boldsymbol{\chi};\Theta)\mu_{\boldsymbol{\chi}}(\Theta,\mathbf{P}) \frac{\nabla_{\Theta_k}\mu_{\chi_k}(\Theta_k,\mathbf{P}_k)}{\mu_{\chi_k}(\Theta_k,\mathbf{P}_k)}
q(\boldsymbol{\chi},\mathbf{P};\Theta),
\end{equation}
where
\begin{equation}
q(\boldsymbol{\chi},\mathbf{P};\boldsymbol{\gamma},\Theta)=g_\psi(\boldsymbol{\chi},\mathbf{P})-\psi(\Theta)+\sum_{\boldsymbol{\chi}^{'}} \Pr\{\boldsymbol{\chi}^{'}| \boldsymbol{\chi},\mathbf{P}\}V( \boldsymbol{\chi}^{'} ).
\end{equation}

Taking the conditional expectation (conditioned on $\{\mathbf{Q},\mathbf{E}\}$) on both sides of \eqref{eq:bellV}, we have following equivalent Bellman equation
\begin{equation}
\widetilde{V}(\mathbf{Q},\mathbf{E})+\psi(\Theta) = \mathbb{E}\left[g_\psi( \boldsymbol{\chi},\Theta)\right] + \sum_{\mathbf{Q}^{'},\mathbf{E}^{'}}\mathbb{E}\left[\Pr\{ \mathbf{Q}^{'},\mathbf{E}^{'} | \boldsymbol{\chi},\Theta\}\right]\widetilde{V}(\mathbf{Q}^{'},\mathbf{E}^{'}),
\end{equation}
where $\widetilde{V}(\mathbf{Q},\mathbf{E})=\mathbb{E}[ V( \boldsymbol{\chi}|\mathbf{Q},\mathbf{E}) ]$. It can be verified that the following differential utility $\widetilde{V}(\mathbf{Q},\mathbf{E})$ of state $(\mathbf{Q},\mathbf{E})$ satisfying the above equivalent Bellman equation
\begin{equation}
\widetilde{V}(\mathbf{Q},\mathbf{E})=\mathbb{E}^{\Omega^{\Theta}}\left[ \sum_{t=0}^{T^I-1}( g_\psi( \boldsymbol{\chi},\Theta)- \psi(\Theta) )  | \mathbf{Q}(0)=\mathbf{Q},\mathbf{E}(0)=\mathbf{E} \right],
\end{equation}
where $T^I=\min\{ t>0| \mathbf{Q}^t=\mathbf{Q}^I,\mathbf{E}^t=\mathbf{E}^I\}$ is the first future time that reference state $(\mathbf{Q}^I,\mathbf{E}^I)$ is visited. Therefore, we have
\begin{equation}
\begin{array}{lll}
&q(\boldsymbol{\chi},\mathbf{P};\Theta)=g_\psi(\boldsymbol{\chi},\mathbf{P})-\psi(\Theta)+\sum_{\mathbf{Q}^{'},\mathbf{E}^{'}} \Pr\{\mathbf{Q}^{'},\mathbf{E}^{'} | \boldsymbol{\chi},\mathbf{P}\}\widetilde{V}(\mathbf{Q}^{'},\mathbf{E}^{'})\\
=&\mathbb{E}^{\Omega^{\Theta}}\Big[ \sum_{t=0}^{T^I-1}\big(g_\psi(\boldsymbol{\chi}(t),\mathbf{P}(t))- \psi(\Theta)\big)| \boldsymbol{\chi}(0)=\boldsymbol{\chi},\mathbf{P}(0)=\mathbf{P} \Big],
\end{array}
\end{equation}
which finishes the proof.

\section{Proof of Lemma \ref{lem:potential_converge}}\label{app:potential_converge}

In timescale I, we can rewrite the update equations in \eqref{eq:learning} as follows
\begin{equation}
r^{t+1}=r^{t}+a(t)R(\boldsymbol{\chi}^t,\mathbf{P}^t,r^{t})
\end{equation}
where $r^t=(\Theta^t,\widetilde{L}^t)$, and
\begin{equation}
R(\boldsymbol{\chi}^t,\mathbf{P}^t,r^{t})=\left[ -(g_\psi(\boldsymbol{\chi}^t,\mathbf{P}^t)-\widetilde{L}^t)z_k^t \atop  -( g_\psi(\boldsymbol{\chi}^t,\mathbf{P}^t) - \widetilde{L}^t  ) \right].
\end{equation}

Define $t_m$ the $m$-th time that the recurrent state $(\mathbf{Q}^I,\mathbf{E}^I)$ is visited, and we have
\begin{equation}
r^{t_{m+1}}=r^{t_m}+\sum_{t=t_m}^{t_{m+1}-1}a(t)R(\boldsymbol{\chi}^t,\mathbf{P}^t,r^{t})
=r^{t_m}+\widetilde{a}(m)h(r^{t_m})+\varepsilon^{m}
\end{equation}
where $\widetilde{a}(m)=\sum_{t=t_m}^{t_{m+1}-1}a(t)$, $\varepsilon^{m}=\sum_{t=t_m}^{t_{m+1}-1}a(t)( R(\boldsymbol{\chi}^t,\mathbf{P}^t,r^{t})- h(r^{t_m}) )$, and
\begin{equation}
h(r^{t_m})=
\left[ -\nabla_{\Theta}\psi(\Theta^{t_m})-(\psi(\Theta^{t_m})-\widetilde{L}^{t_m})
W(\Theta^{t_m})/\mathbb{E}^{\Omega^{\Theta^{t_m}}}[T^I] \atop -(\psi(\Theta^{t_m})-\widetilde{L}^{t_m}) \right],
\end{equation}
where $W(\Theta)=[W_1(\Theta),\cdots,W_K(\Theta)]$, and $W_k(\Theta)=\mathbb{E}^{\Omega^{\Theta}}\big[\sum_{t=t_m}^{t_{m+1}-1}(t_{m+1}-t) \frac{\nabla_{\Theta_k}\mu_{\chi_k}(\Theta_k,\mathbf{P}_k)}{\mu_{\chi_k}(\Theta_k,\mathbf{P}_k)}  \big]$.
Next we shall show that the following holds almost surely.
\begin{equation}
\sum_{m=1}^{\infty}\widetilde{a}(m)=\infty,\quad\sum_{m=1}^{\infty}[\widetilde{a}(m)]^2<\infty.
\end{equation}

Specifically, $\sum_{m=1}^{\infty}\widetilde{a}(m)=\sum_{t=1}^{\infty}a(t)=\infty$. Furthermore, since $a(t)$ is non-increasing and $(\mathbf{Q}^I,\mathbf{E}^I)$ is a recurrent state, we have
\begin{equation}
\begin{array}{ll}
\mathbb{E}[\sum_{m=1}^{\infty}[\widetilde{a}(m)]^2]\leq \mathbb{E}[\sum_{m=1}^{\infty}a(t_m)^2(t_{m+1}-t_{m})^2] =\mathbb{E}[\sum_{m=1}^{\infty}a(t_m)^2\mathbb{E}(t_{m+1}-t_{m})^2]<\infty.
\end{array}
\end{equation}

Therefore, $\sum_{m=1}^{\infty}[\widetilde{a}(m)]^2$ has finite expectation and is finite almost surely. Following the same way, it is easy to infer that $\sum_{m}\varepsilon^{m}$ has finite expectation and is finite almost surely, and hence $\sum_{m}\varepsilon^{m}$ converges almost surely. Since $\varepsilon_m$ and $\widetilde{a}(m)$ converges to zero almost surely, and $h(r^{t_m})$ is bounded, we have
\begin{equation}
\lim_{m\to\infty}(r^{t_{m+1}}-r^{t_m})=0.
\end{equation}

Then, similar to the proof of \cite[Lemma 11]{Marbach:TAC:2001}, we can show that $\psi(\Theta^{t_m})$ and $\widetilde{L}^{t_m}$ converge to a common limit. Since $\psi(\Theta^{t_m})-\widetilde{L}^{t_m}$ converges to zero, the algorithm of the per-user parameter update is given by
\begin{equation}
\Theta^{t_{m+1}}=\Theta^{t_m}+\widetilde{a}(m)(\nabla \psi(\Theta^{t_m})+e^m)+\varepsilon^{m},
\end{equation}
where $e^m$ converges to zero and $\varepsilon^{m}$ is a summable sequence. This is a gradient method with diminishing errors. Therefore, by following the same way as in \cite{Borkar:2008} and \cite{Bertsekas:gradient:2000}, we can conclude that the learning algorithm will converges to a equilibrium $\Theta^{\infty}$ almost surely, given by
\begin{equation}
\nabla_{\Theta}\psi(\Theta^{\infty},\boldsymbol{\gamma})=0. 
\end{equation}

\section{Proof of Lemma \ref{lem:lm_converge}}\label{app:lm_converge}
Due to the separation of timescale, the primal update of the
per-user parameter can be regarded as converged to
$\Theta_k^*(\boldsymbol{\gamma}^t)$ w.r.t. the current LMs
$\boldsymbol{\gamma}^t$. Specifically, for timescale II, we can rewrite the update equations in \eqref{eq:learning} as follows
\begin{equation}
\gamma_k^{t+1}=\Big[ \gamma_k^{t+1}+b(t)\big( \overline{P}_k(\Theta^*(\boldsymbol{\gamma}))-P_k^0 +
\underbrace{P_{k,ac}-\overline{P}_k(\Theta^*(\boldsymbol{\gamma}))}_{w_k^{t+1}} \big)\Big]^+.
\end{equation}

Let $\mathbb{F}^t=\sigma(\boldsymbol{\gamma}^{l},w_k^{l},l\leq t)$ be the $\sigma$-algebra generated by $\{\gamma_k^{l},w_k^{l},l\leq t\}$. Note that $\mathbb{E}[w_k^{t+1}|\mathbb{F}^t] = 0$, and $\mathbb{E}[||w_k^{t+1}||^2|\mathbb{F}^t] \leq C_1(1+||\boldsymbol{\gamma} ||)$ for a suitable constant $C_1$. Using the standard stochastic approximation
argument \cite{Borkar:2008}, the dynamics of the LMs learning
equation in \eqref{eq:learning} for user $k$ can be represented by the
following ordinary differential equation (ODE):
\begin{equation}\label{eq:gamma_ODE}
\dot{\gamma_k}(t)=\overline{P}_k(\Theta^*(\boldsymbol{\gamma}(t)))-P_k^0,
\end{equation}
where
$\Theta^*(\boldsymbol{\gamma}(t))$ is the converged per-user parameter under the LM $\boldsymbol{\gamma}^t$. Define
\begin{equation}
G(\boldsymbol{\gamma})=\psi(\Theta^*(\boldsymbol{\gamma}),\boldsymbol{\gamma})=\sum_k\left(\beta_k\overline{T}_k(\Theta^*(\boldsymbol{\gamma}))+ \gamma_k(\overline{P}_k(\Theta^*(\boldsymbol{\gamma}))-P_k^0) \right).
\end{equation}

By the chain rule and $\nabla_{\Theta}\psi(\Theta^*(\boldsymbol{\gamma}),\boldsymbol{\gamma})=\mathbf{0}$, we have $\frac{\partial G(\boldsymbol{\gamma})}{\partial \gamma_k}=\sum_{k,i}\frac{\partial \psi(\Theta,\boldsymbol{\gamma})}{\partial \theta_{ki}}\frac{\partial \theta_{ki}}{\partial \gamma_k}\big|_{\Theta=\Theta^*(\boldsymbol{\gamma})}+\frac{\partial \psi(\Theta,\boldsymbol{\gamma})}{\partial \gamma_k}=\overline{P}_k(\Theta^*(\boldsymbol{\gamma}))-P_k^0$. Therefore, we show that the ODE in \eqref{eq:gamma_ODE} can be
expressed as $\dot{\gamma_k}(t) = \frac{\partial G(\boldsymbol{\gamma}(t))}{\partial \gamma_k}$. As a result, the ODE in
\eqref{eq:gamma_ODE} will either converge to $\overline{P}_k(\Theta^*(\boldsymbol{\gamma}))-P_k^0=0$ or $\{\frac{\partial G(\boldsymbol{\gamma}(t))}{\partial \gamma_k}<0,\gamma_k=0\}$ which satisfies the average power constraints in \eqref{eq:pwr_con} .

\section{Proof of Theorem \ref{thm:posg}}\label{app:thm:posg}
Note that when $\eta=0$, i.e., the interference is zero for each user, we have $\overline{T}_k(\Theta_k,\Theta_{-k};\eta=0)=\overline{T}_k(\Theta_k)$ and $\overline{P}_k(\Theta_k,\Theta_{-k};\eta=0)=\overline{P}_k(\Theta_k)$ in problem \eqref{eq:problem:POSG}. In other words, the other users' control policies $\Omega_{-k}^{\Theta_{-k}}$ do not influence the average delay utility $\overline{T}_k$ and AC power consumption $\overline{P}_k$ for user $k$, since there is no interference. Specifically, denote $\overline{T}_k^0(\Theta_k)=\overline{T}_k(\Theta_k,\Theta_{-k};\eta=0)$ $\overline{P}_k^0(\Theta_k)=\overline{P}_k(\Theta_k,\Theta_{-k};\eta=0)$, and hence $\psi_k^0(\Theta_k,\gamma_k)=\psi_k(\Theta_k,\Theta_{-k},\gamma_k;\eta=0)$. From the convergence analysis in Section \ref{sec:pomdp_convergence}, the per-user parameter and LM will converges to the equilibrium point $\Theta_k^0=\Theta_k^*$ given by
\begin{equation}
\nabla_{\Theta_k}\psi_{k}(\Theta_k^0,\Theta_{-k}^0,\gamma_k^0)=0, \nabla_{\Theta_k\Theta_k}\psi_{k}(\Theta_k^0,\Theta_{-k}^0,\gamma_k^0)\succ\mathbf{0}
\end{equation}
and $\gamma^0$ satisfies the AC power constraint.

Let $\lambda_k=[\Theta_k;\gamma_k]$ and $\Lambda=[\lambda_1,\cdots,\lambda_K]$. From \cite{Borkar:2008},
we can rewrite the update algorithm as the following ODE
\begin{equation}
\dot{\Lambda}(t)=f^0(\Lambda(t))=\left[ \left.
 -\nabla_{\Theta_1}\psi_1^0(\Theta_1,\gamma_1) \atop \overline{P}_1^0(\Theta_1)-P_1^0 \right. ;
 \cdots; \left. -\nabla_{\Theta_K}\psi_K^0(\Theta_K,\gamma_K) \atop \overline{P}_K^0(\Theta_K)-P_K^0\right.
\right].
\end{equation}

Note that $f^0(\Lambda^0)=0$, i.e., $\Lambda^0$ is the equilibrium point for the above ODE. The Jacobian matrix $Df^0(\Lambda^0)$ at the equilibrium point $\Lambda_0$ is given by
\begin{equation}
Df^0(\Lambda^0)=\left(
\begin{array}{ccc}
\begin{array}{ccc}
-\nabla_{\Theta_1\Theta_1}^2\psi_1^0(\Theta_1,\gamma_1) & -\nabla_{\Theta_1}\overline{P}_1^0(\Theta_1)  \\
\nabla_{\Theta_1}\overline{P}_1^0(\Theta_1) &  \mathbf{0}
\end{array} & \cdots \\
\cdots &
\begin{array}{ccc}
-\nabla_{\Theta_K\Theta_K}^2\psi_K^0(\Theta_K,\gamma_K) & -\nabla_{\Theta_K}\overline{P}_K^0(\Theta_1)  \\
\nabla_{\Theta_K}\overline{P}_K^0(\Theta_K) &  \mathbf{0}
\end{array}
  \end{array}
  \right).
\end{equation}
Since $\nabla_{\Theta_K\Theta_K}^2\psi_K^0(\Theta_K,\gamma_K)\succ\mathbf{0}$, it has been shown in \cite{ber:book:1999} that all eigenvalues of $Df^0(\Lambda^0)$ have strictly negative real parts. Therefore, $\Lambda^0$ is exponentially stable. By converge of Lyapunov Theorem \cite{Lyapunov:1998}, there exists a Lyapunov function $V(\Lambda)$ for $\dot{\Lambda}(t)=f^0(\Lambda(t))$, s.t.
$C_1||\Lambda-\Lambda^0||^2\leq V(\Lambda)\leq C_2||\Lambda-\Lambda^0||^2 $, and $\frac{\text{d}V(\Lambda)}{\text{d}\Lambda}
f^0(\Lambda)\leq-C_3||\Lambda-\Lambda^0||^2,\forall \Lambda $ for some positive constant $\{C_1,C_2,C_3\}$.

When $\eta\neq0$, let $\eta_k=\max_i\frac{L_{ki}}{L_{kk}}$, and from Taylor expansion we have
\begin{eqnarray}
\psi_k(\Theta,\gamma_k)&=&\psi_k^0(\Theta_k,\gamma_k)+\sum_{i=1}^K\frac{L_{ki}}{L_{kk}}\frac{\partial \psi_k^0(\Theta_k,\gamma_k)}{\partial L_{ki}/L_{kk}}+O(\eta_k^2) \\
\overline{P}_k(\Theta,\gamma_k)&=&\overline{P}_k^0(\Theta_k,\gamma_k)+\sum_{i=1}^K\frac{L_{ki}}{L_{kk}}\frac{\partial \overline{P}_k^0(\Theta_k,\gamma_k)}{\partial L_{ki}/L_{kk}}+O(\eta_k^2)
\end{eqnarray}
Therefore, we can rewrite the update algorithm as the following ODE
\begin{equation}
\dot{\Lambda}(t)=f^0(\Lambda(t))+\epsilon(\Lambda(t))
\end{equation}
where $||\epsilon(\Lambda(t))||=O(\eta)$. Then we have
\begin{equation}\begin{array}{lll}
\dot{V}(\Lambda)\triangleq
\frac{\text{d}V}{\text{d}t}&=&\frac{\text{d}V}{\text{d}\Lambda}\dot{\Lambda}
=\frac{\text{d}V}{\text{d}\Lambda}(
f^0(\Lambda)+\epsilon(\Lambda)
) \leq-C_3||\Lambda-\Lambda^0||^2+2
C_2||\Lambda-\Lambda^0||\cdot||\epsilon(\Lambda)||\\
&=&-||\Lambda-\Lambda^0||(C_3||\Lambda-\Lambda^0||-
2 C_2||\epsilon(\Lambda)|| ).
\end{array}.
\end{equation}

Note that $\dot{V}(\Lambda)<0$ for all
$\Lambda$ s.t.
$(C_3)^2||\Lambda-\Lambda^0||^2\geq
4(C_2)^2||\epsilon(\Lambda)||^2=\delta=O(\eta^2)$.
As a result, $\Lambda^t$ converges almost surely to an
invariant set given by $\mathcal{S}\triangleq
\left\{\Lambda:
||\Lambda-\Lambda^0||^2-\delta \leq0
\right\}$. Furthermore, from $\dot{V}(\Lambda^*)=0$, we
have $||\Lambda^*-\Lambda^0||^2-\delta
\leq0$, and hence the invariance set is also given by $
\mathcal{S}\triangleq \left\{\Lambda:
||\Lambda-\Lambda^*||^2-\delta
\leq0 \right\}$. Finally, we can conclude that the iterations of the per-user parameter
$\Theta^t$ in the proposed learning algorithm \ref{alg:learning_posg} will converge almost surely to an invariant set given by
\begin{equation}
\mathcal{S}_{\theta}\triangleq\{ \Theta: ||\Theta - \Theta^*||-\delta\leq0\}
\end{equation}
as $t\to\infty$, for some positive constant $\delta=O(\eta^2)$ and some $\Theta^*\in\mathcal{F}^*$.

\bibliographystyle{IEEEtran}
\bibliography{IEEEabrv,POSG}


{
\begin{table}[h]
\begin{center}
\caption{ Communication overhead comparison for exchanging the per-stage utility and sharing the buffer states and the CSI states}{
\begin{tabular}{|c|c|c|}
\hline & Exchanging the per-stage utility & Sharing the buffer states and the CSI states \\
\hline communication overhead & $O\big(\prod_k(N_k^QN_k^E)\big)$ & $O\big(\prod_k(N_k^QN_k^E)\prod_{k,n}N_{kn}^{H}\big)$ \\
\hline
\end{tabular}}
\label{tab:comp}
\end{center}
\end{table}
}

\begin{figure}[h]
 \begin{center}
  \resizebox{14cm}{!}{\includegraphics{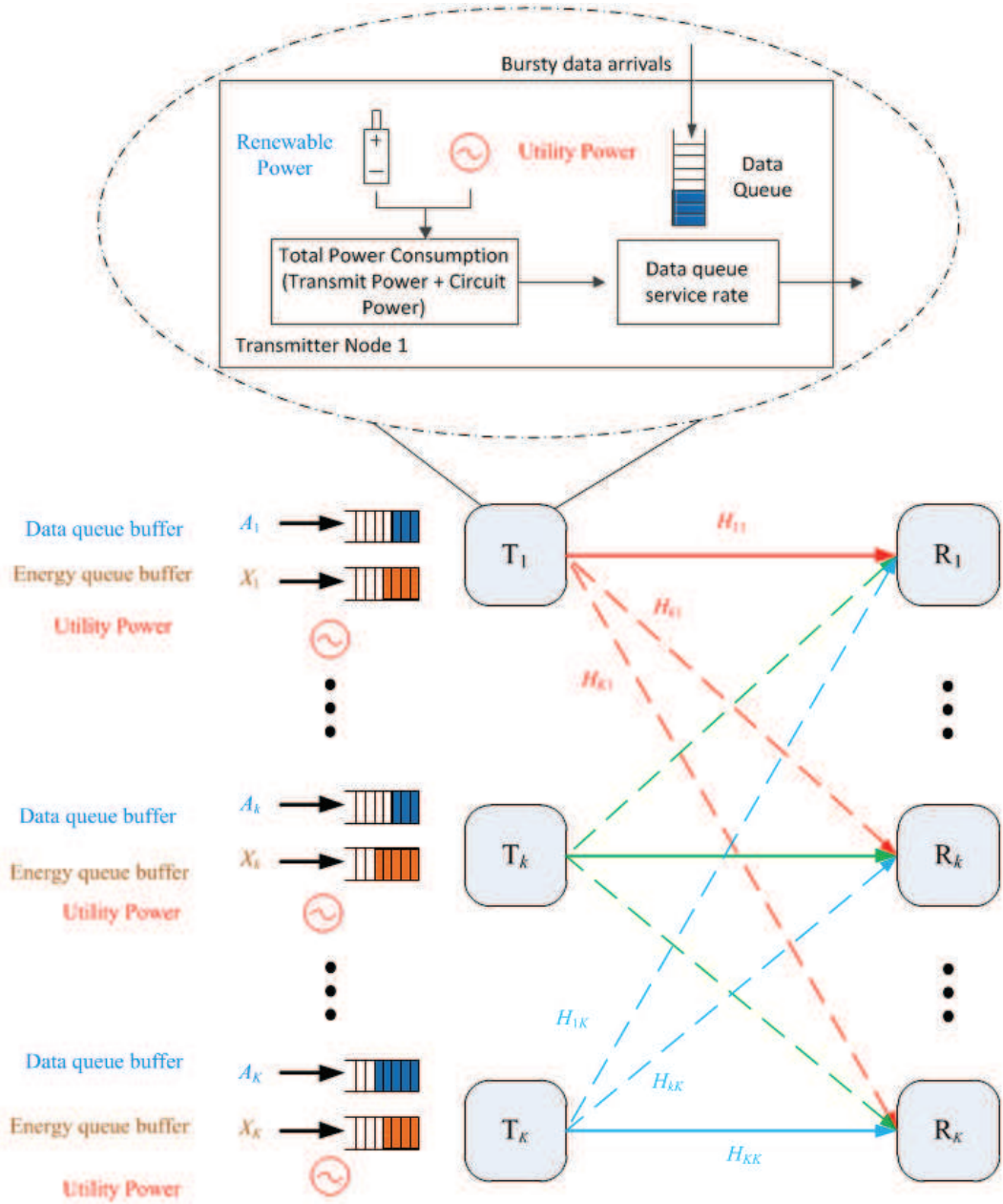}}
 \end{center}
    \caption{
System model. Each transmitter maintains a data queue for the random traffic flow towards the desired receiver in the system.
Furthermore, each transmit node is capable of harvesting energy from the environment and storing it in an energy buffer.}
    \label{fig:system_model}
\end{figure}

\begin{figure}
 \begin{center}
  \resizebox{14cm}{!}{\includegraphics{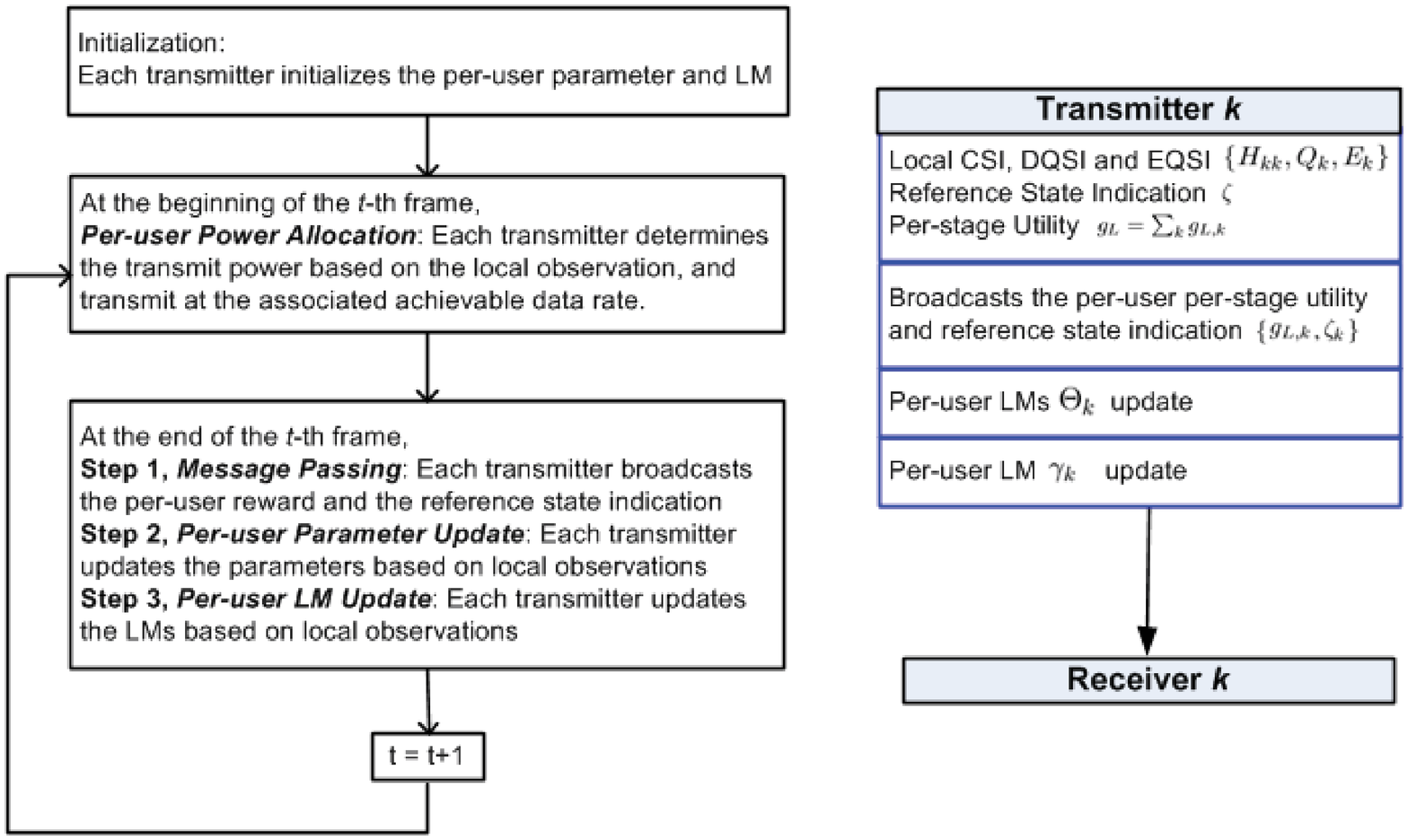}}
 \end{center}
    \caption{
The system procedure of the decentralized per-user parameter and LM online learning algorithm for DEC-POMDP problem.}
    \label{fig:learning_pomdp}
\end{figure}

\begin{figure}
 \begin{center}
  \resizebox{14cm}{!}{\includegraphics{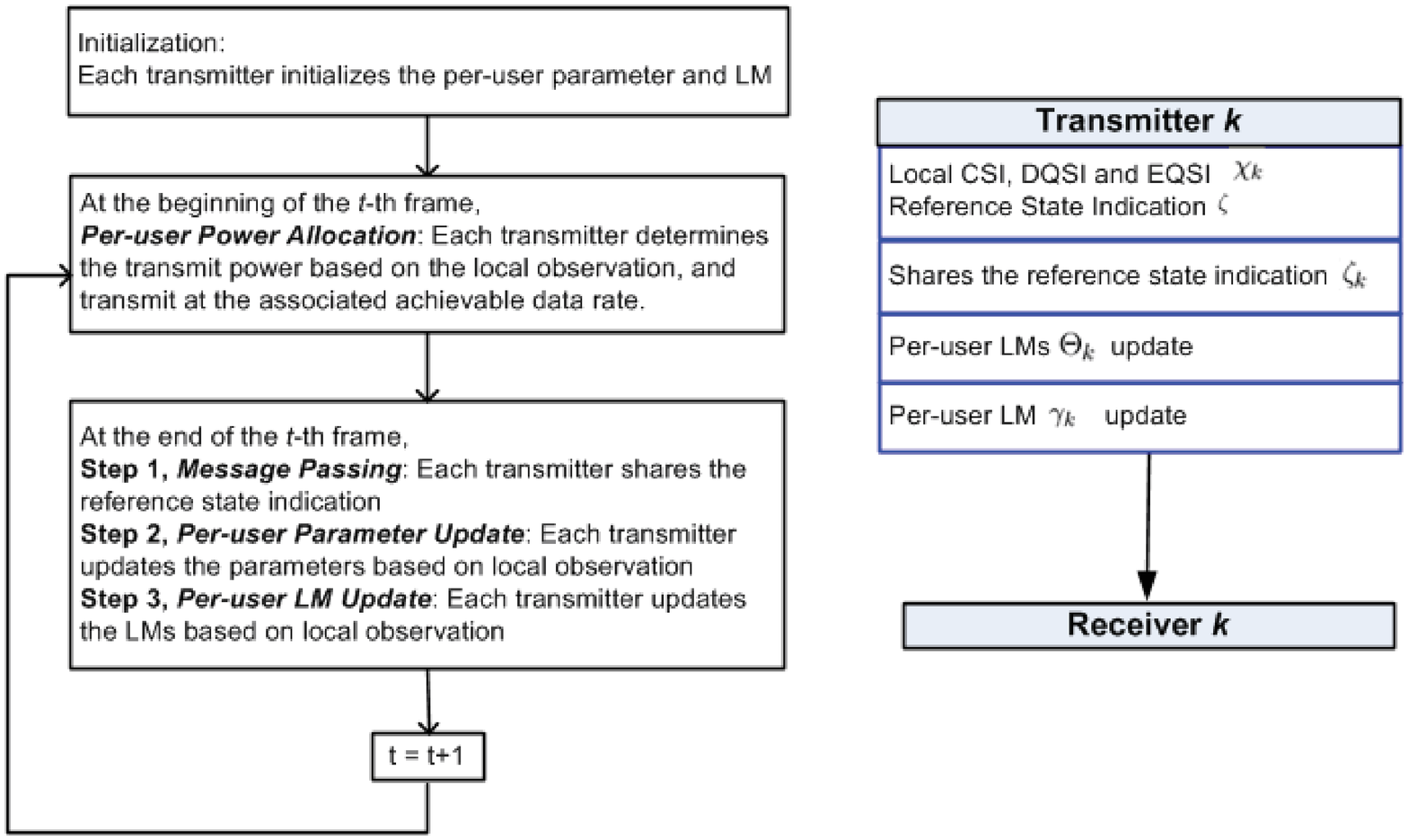}}
 \end{center}
    \caption{
The system procedure of the decentralized per-user parameter and LM online learning algorithm for POSG problem.}
    \label{fig:learning_posg}
\end{figure}

\begin{figure}
 \begin{center}
  \resizebox{11cm}{!}{\includegraphics{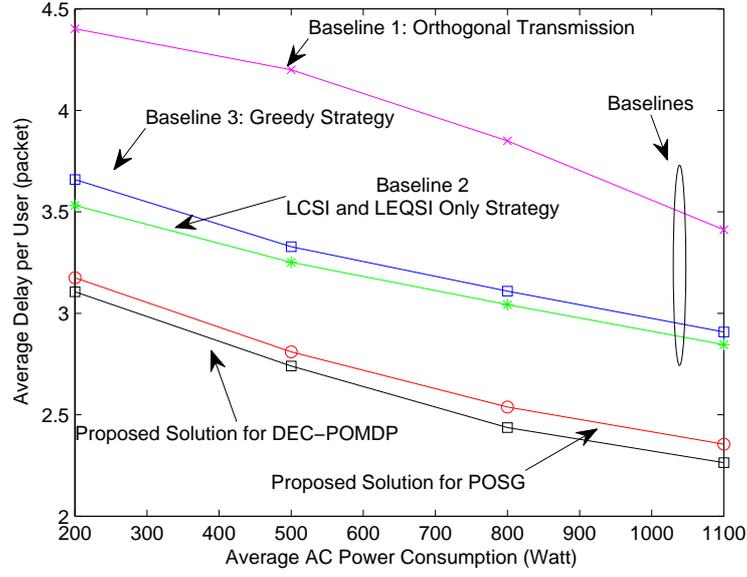}}
 \end{center}
    \caption{
Delay performance per user versus the AC power consumption $P_k^0$. The average data arrival rate is $\lambda_k=1.1$ (packet per second), and energy arrival rate is $\overline{X}_k=800$ (Watt). The renewable energy is stored in a 1.2V 20Ah lithium-ion battery. }
    \label{fig:acpower}
\end{figure}

\begin{figure}
 \begin{center}
  \resizebox{11cm}{!}{\includegraphics{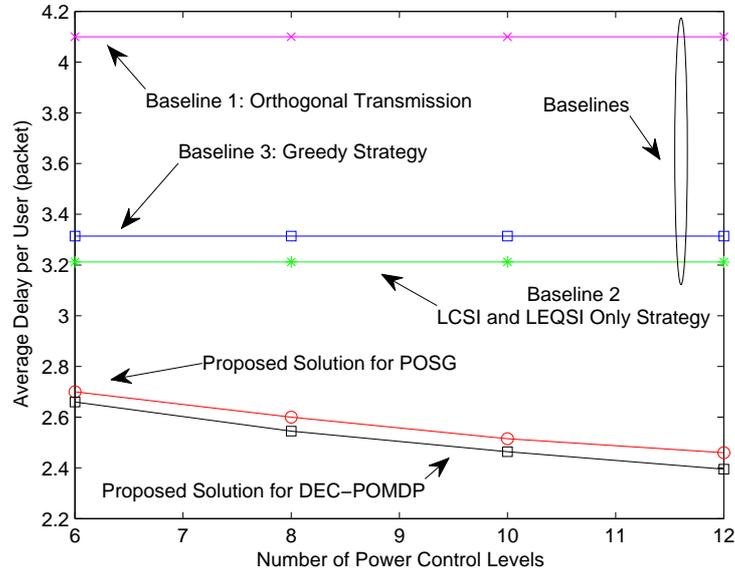}}
 \end{center}
    \caption{
Delay performance per user versus the number of power control levels that lie in 0 and 1.5kW. The average data arrival rate is $\lambda_k=1.1$ (packets per second), the energy arrival rate is $\overline{X}_k=800$ (Watt), and the average AC power consumption is $P_k^0=800$ (Watt). The renewable energy is stored in a 1.2V 20Ah lithium-ion battery. }
    \label{fig:rate}
\end{figure}

\begin{figure}
 \begin{center}
  \resizebox{14cm}{!}{\includegraphics{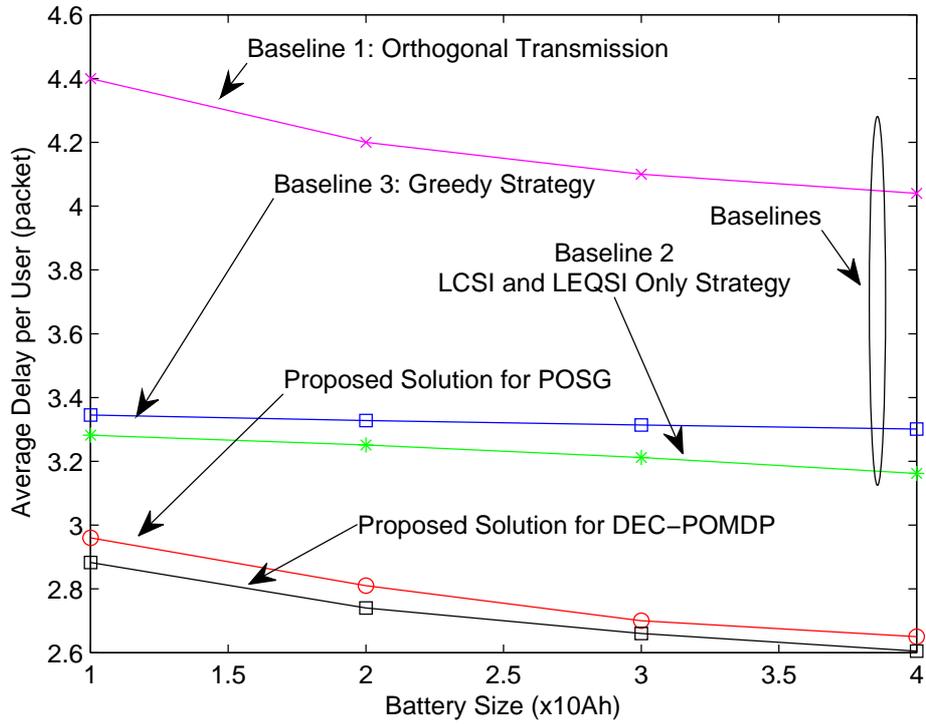}}
 \end{center}
    \caption{
Delay performance per user versus the renewable energy buffer size $N_{k}^{E}$. Specifically, we consider the lithium-ion battery given from 1.2V 10Ah to 40Ah. The average data arrival rate is $\lambda_k=1.1$ (packets per second), the energy arrival rate is $\overline{X}_k=800$ (Watt), and the average AC power consumption is $P_k^0=500$ (Watt). }
    \label{fig:buffer}
\end{figure}

\begin{figure}
\begin{center}
  \subfigure[POMDP Problem]
  {\resizebox{8cm}{!}{\includegraphics{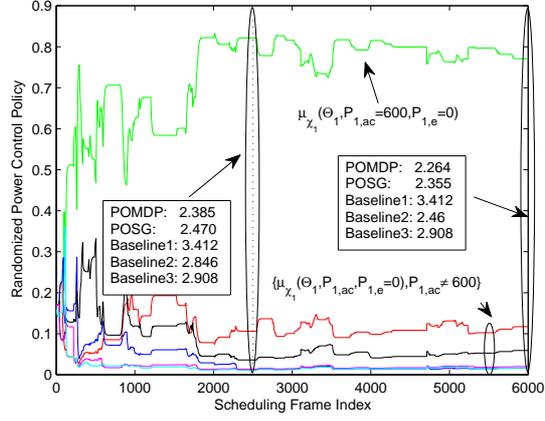}}}
  \subfigure[POSG Problem]
  {\resizebox{8cm}{!}{\includegraphics{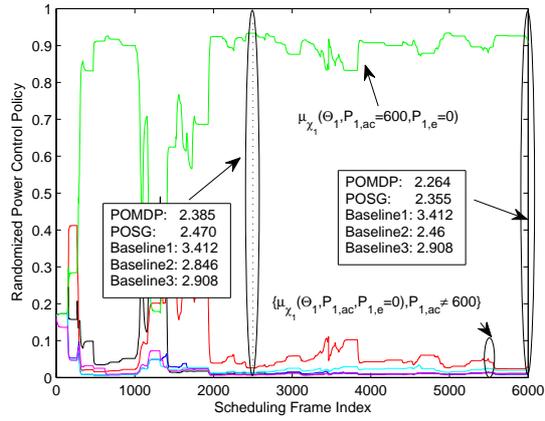}}}
  \end{center}
    \caption{ Convergence property of the proposed scheme. The average data arrival rate is $\lambda_k=1.1$ (packets per second), the energy arrival rate is $\overline{X}_k=800$ (Watt), and the average AC power consumption is $P_k^0=1100$ (Watt). The figure illustrates the instantaneous randomized power control policy $\mu_{\chi_1}(\Theta_1,\mathbf{P}_1)$ ($Q_1=2,E_1=0$) versus scheduling frame index for the POMDP and POSG problems, respectively. The boxes indicate the average delay of various schemes at the selected frame indices.}
    \label{fig:convergence}
\end{figure}

\end{document}